\documentclass[pra,twocolumn,amsmath,amssymb,superscriptaddress]{revtex4-2}
\usepackage[utf8]{inputenc}
\usepackage[dvipsnames]{xcolor}
\usepackage{amsfonts}
\usepackage{amsmath}
\usepackage{amssymb}
\usepackage{amsthm}
\usepackage{svg}
\usepackage{qcircuit}
\usepackage{xfrac} 
\usepackage{graphicx}
\usepackage{braket}
\usepackage{lipsum}
\usepackage{ORCIDinREVTeX}

\usepackage{bm}
\usepackage[makeroom]{cancel}

\definecolor{myred}{rgb}{0.8500,0.3250,0.0980}
\definecolor{mygreen}{rgb}{0.4660,0.6740,0.1880}
\usepackage{tikz}
\usetikzlibrary{positioning}

\usepackage{comment}

\newtheorem{theorem}{Theorem}
\newtheorem{definition}{Definition}

\newcommand{\hc}{\text{H.c.}}
\newcommand{\bblock}{$B$-block}
\newcommand{\pblock}{$P$-block}

\newcommand{\qblock}{$Q$-block}
\newcommand{\qdiam}{diamond}

\newcommand{\CTFIM}{\mathrm{CTFIM}}
\newcommand{\TFIM}{\mathrm{TFIM}}
\newcommand{\PTFIM}{\mathrm{PTFIM}}
\newcommand{\TFXY}{\mathrm{TFXY}}

\newcommand{\ham}{\mathcal{H}}
\newcommand{\fswap}{\mathcal{F}}

\newcommand{\bigO}{\mathcal{O}}

\newcommand{\redtext}[1]{\textcolor{black}{#1}}
\newcommand{\rredtext}[1]{\textcolor{black}{#1}}
\newcommand{\bluetext}[1]{\textcolor{blue}{#1}}

\newcommand{\myvec}{\Vec}
\usepackage[colorlinks,citecolor=blue,linkcolor=blue]{hyperref}
\usepackage[capitalise,nameinlink]{cleveref}
\crefname{section}{Sec.}{Secs.}
\crefname{table}{Tab.}{Tabs.}
\crefname{figure}{Fig.}{Figs.}
\crefname{definition}{Def.}{Defs.}
\crefname{lema}{Lem.}{Lems.}
\crefname{theorem}{Thm.}{Thms.}
\crefname{corollary}{Cor.}{Cors.}

\begin{document}

\title{Algebraic Compression of Free Fermionic Quantum Circuits:\texorpdfstring{\\}{}Particle Creation, Arbitrary Lattices and Controlled Evolution}

\author{Efekan K\"okc\"u}
\orcid{0000-0002-7323-7274}
\email{ekokcu@lbl.gov}
\affiliation{Applied Mathematics and Computational Research Division,
            Lawrence Berkeley National Laboratory,
            Berkeley, CA 94720, USA}
\affiliation{Department of Physics, North Carolina State University, Raleigh, North Carolina 27695, USA}

\author{Daan Camps}
\orcid{0000-0003-0236-4353}
\affiliation{National Energy Research Scientific Computing Center,
            Lawrence Berkeley National Laboratory,
            Berkeley, CA 94720, USA}
            
\author{Lindsay Bassman Oftelie}
\orcid{0000-0003-3542-1553}
\affiliation{Applied Mathematics and Computational Research Division,
            Lawrence Berkeley National Laboratory,
            Berkeley, CA 94720, USA}

\author{Wibe~A.~de~Jong}
\orcid{0000-0002-7114-8315}
\affiliation{Applied Mathematics and Computational Research Division,
            Lawrence Berkeley National Laboratory,
            Berkeley, CA 94720, USA}
            
\author{Roel~Van~Beeumen}
\orcid{0000-0003-2276-1153}
\affiliation{Applied Mathematics and Computational Research Division,
            Lawrence Berkeley National Laboratory,
            Berkeley, CA 94720, USA}

\author{A.~F.~Kemper}
\orcid{0000-0002-5426-5181}
\email{akemper@ncsu.edu}
\affiliation{Department of Physics, North Carolina State University, Raleigh, North Carolina 27695, USA}

\date{\today}

\begin{abstract}
Recently we developed a local and constructive algorithm based on Lie algebraic methods for compressing Trotterized
evolution under Hamiltonians that can be mapped to free fermions~\cite{kokcu2022algebraic,camps2022algebraic}. The compression algorithm yields a 
circuit which scales linearly in the number of qubits, has a depth independent of evolution time and
compresses time-dependent Hamiltonians. 
The algorithm is limited to simple nearest-neighbor spin interactions and fermionic hopping. In this work, we extend our methods
to compress 
evolution with long-range 
fermionic hopping, thereby enabling the embedding of arbitrary lattices onto a chain of qubits for fermion models. Moreover, we show that controlled time evolution, as well as fermion creation and annihilation operators can also be compressed. We demonstrate our results by adiabatically preparing the ground state for a half-filled fermionic chain, simulating
a $4 \times 4$ tight binding model on \emph{ibmq\_washington}, and calculating the topological Zak phase on a Quantinuum H1-1 trapped-ion quantum computer.
With these new developments, our results enable the simulation of a wider range of models of interest and the efficient compression of subcircuits.
\end{abstract}

\maketitle

\section{Introduction}\label{sec:intro}%

Time evolution of a quantum state is one of the potential areas where quantum computers are expected to have an advantage over classical computers;
while state preparation lies in the QMA complexity class, time evolution is part of the simpler BQP class. 
Importantly for applications in physics and chemistry, time evolution is a building block for quantum simulation of physical and chemical systems.
It underpins state preparation via adiabatic time evolution, and is a critical step in determining dynamic response functions~\cite{chiesa2019quantum,Roggero19,francis2020quantum,Kosugi20,Kosugi20linear,endo2020calculation,libbi2022effective, kokcu2023linear, steckmann2023mapping, joven2024out}.
Furthermore, time evolution under a piecewise
constant Hamiltonian has also garnered some recent interest in the context of Floquet dynamics and phase transitions~\cite{rodriguez2022real,zhang2022digital,mi2022time}.

In all cases, a crucial challenge remains: How do we implement the time evolution operator $U(t)$ on a digital quantum computer? 
The Hamiltonian is typically a sum of Pauli strings $\sigma_k$,
\begin{align}\label{eq:intro_ham}
    \ham(t) = \sum_k \alpha_k(t) \sigma_k,
\end{align} 
where $\alpha_k(t)$ are real functions of time $t$, and $U(t)$ is related to the Hamiltonian $\ham(t)$ via the following equation
\begin{align}
    \partial_t U(t) = -i \ham(t) U(t),
\end{align}
with initial condition $U(t=0) = \mathcal{I}$.
In the simplest case where the Hamiltonian $\ham$ is time-independent, this expression simplifies to $U(t) = \exp(-i\ham t)$.
However, the Pauli strings $\sigma_k$ in the Hamiltonian \eqref{eq:intro_ham} usually do not commute with each other, and thus
$U(t)$ cannot be
directly decomposed into the 1- and 2-qubit gates that are available on a quantum computer.
Finding such a decomposition for $U(t)$ is entitled unitary synthesis and is a difficult problem that generally scales exponentially in the
number of qubits~\cite{Shende,khaneja2005constructive,drury2008quantum_shannon}.

A natural choice for decomposition, 
which covers the time-dependent case as well,
is by use of Trotterization, or expansion of the time evolution operator into small time steps $\delta t$,
\begin{align}
    U(t) \approx \underbrace{e^{-i \ham(t_1) \delta t} e^{-i \ham(t_2) \delta t} \cdots e^{-i \ham(t_N) \delta t}}_{N\, \mathrm{ time\:steps}} ,
\end{align}
where $t_n = n\delta t$ and $\delta t = t/N$.
This has the added advantage that the individual small time steps can be broken up to separate the Pauli strings while limiting the approximation error
\begin{align}
    e^{-i\ham(t_n) \delta t} \approx \prod_k e^{-i\alpha_k(t_n) \sigma_k \delta t} + \mathcal{O}(\delta t^2),
\end{align}
which leads to a natural circuit implementation.  The downside of this approach is that the circuit depth grows with simulation time $t$,
which is undesirable for current era noisy quantum computers which are limited in terms of circuit depths that can be run reliably.

There are a variety of methods proposed to  overcome or bypass the problems that arise from the deep circuits from a na\"ive Trotter implementation.
These include variational approaches to learn shorter circuits~\cite{commeau2020variational,bassman2022constant,zhang2024scalable}, and decompositions based on algebraic methods
that do not grow as a function of simulation time~\cite{khaneja2005constructive,earp2005constructive,drury2008quantum_shannon,kokcu2022fixed,kokcu2022algebraic,camps2022algebraic,wierichs2025recursive}.

In general, due to no-go theorems on fast-forwarding quantum evolution, these circuits may be exponentially deep~\cite{berry2007efficient,atia2017fast,gu2021fast}.
However, for Hamiltonians that can be mapped onto free fermionic ones, shallow fixed depth circuits, that are independent of simulation time, exist~\cite{kokcu2022algebraic,camps2022algebraic,babush1,babush2,babush3,peng2022quantum,pnnl_compression}. 
Previous work by the authors and others have demonstrated that the circuit elements that arise from evolution under a \emph{time dependent} 1D nearest-neighbor free fermionic Hamiltonian obey
three algebraic properties --- fusion, commutation, and turnover --- which can be used for an efficient compression algorithm that results in a fixed-depth circuit.
While the class of free fermionic Hamiltonians is somewhat restrictive, since it includes some of the more common models studied in quantum computing --- such as the nearest-neighbor transverse field Ising model (TFIM) and nearest-neighbor transverse field XY model (TFXY) --- the compression algorithm has found applications~\cite{oftelie2022simulating,kokcu2023linear},
or has otherwise set a benchmark for what comprises a minimal
circuit~\cite{dupont2022entanglement,sopena2022algebraic,jian2023towards,low2022on,kaldenbach2022digital,wierichs2025recursive}.
Since the compression algorithm starts from a Trotter decomposition, it readily handles evolution under a time-dependent Hamiltonian as well.
Moreover, the circuit compression is not limited to time evolution, but is applicable to any circuit elements that obey a few key relationships.
Thus, the compression algorithm can also be used for subcircuits built of free fermionic operators.
This situation arises regularly because part of interacting fermionic evolution is free fermion hopping.

\begin{figure*}[htpb]
    \centering
    \includegraphics[width = 1.65 \columnwidth]{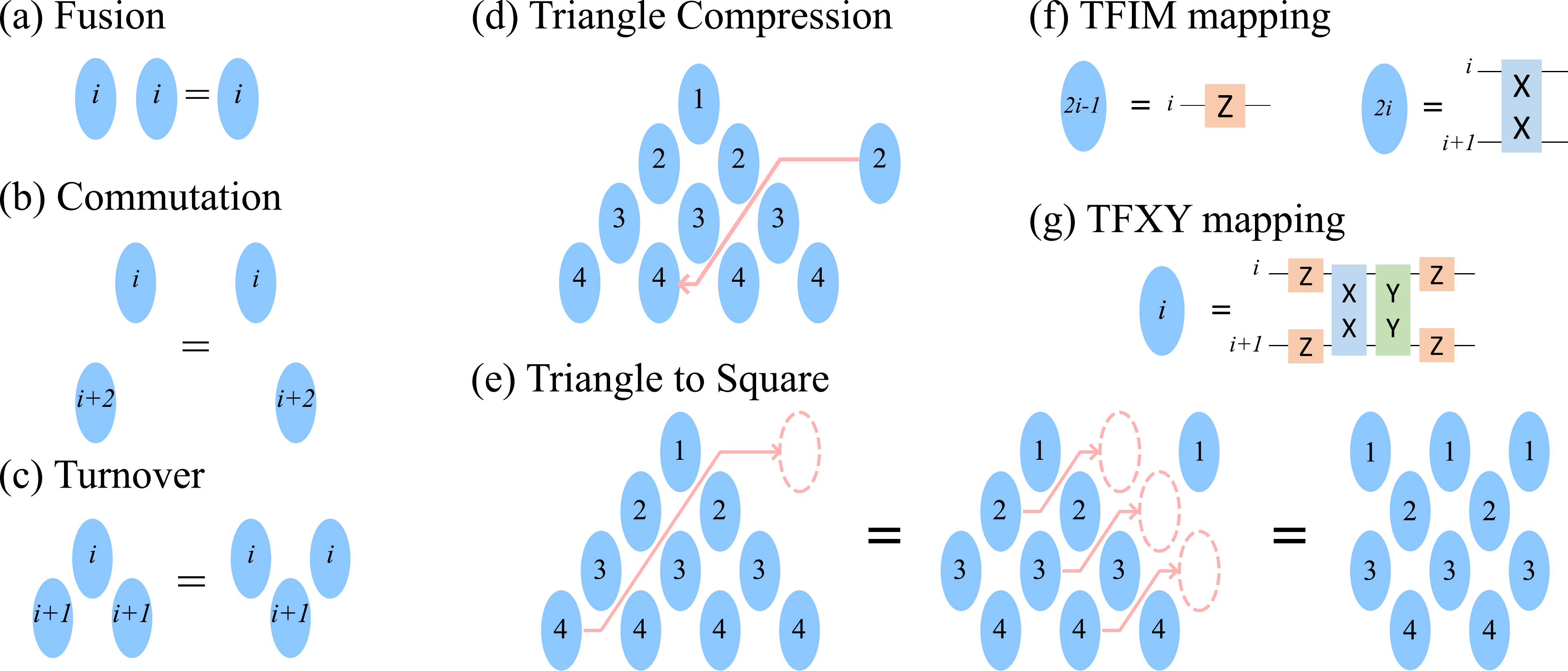}
    \caption{Summary of the results from Refs.~\cite{kokcu2022algebraic,camps2022algebraic} that we will use in this work. Panels (a-c) illustrate the block properties listed in \cref{def:block}. Panel (d) illustrates the triangle structure and how it can absorb a \bblock\ via the \bblock\ properties. Panel (e) is an illustration of the usage of \bblock\ properties to transform a triangle into a shallower square structure. Panels (f) and (g) illustrate two different realizations or representations of the abstract \bblock\ objects as circuit elements, i.e. \bblock mappings, for transverse field Ising model (TFIM) and transverse field XY (TFXY) model. \rredtext{Note that \bblock mappings are not limited to TFIM and TFXY, and more examples can be found in Refs.~\cite{kokcu2022algebraic,camps2022algebraic}.}
    }
    \label{fig:all_block_things}
\end{figure*}

In this paper, we extend our previous compression work in four ways. 
We show that the evolution under the following can be compressed: 
\begin{enumerate}
    \item long-range fermionic hopping and fermionic swap (FSWAP) operators (\cref{sec:swap_and_longrange}),
    \item Hermitian linear combinations of particle creation/annihilation operators (\cref{sec:creation_annihilation}), 
    \item free fermionic evolution controlled by a single ancilla qubit (\cref{subsec:controlled_free_fermions}),
    \item 1D TFIM, TFXY, and Kitaev models with periodic boundary conditions (\cref{asec:pbc}).
\end{enumerate}
These developments, when combined, lead to a number of significant improvements and enhanced capabilities of our approach.
First, long-range hopping and FSWAP compression enables the compression of free fermionic evolution on any lattice, not just linear chains.
Moreover, the FSWAP operator is key in any situation where fermionic statistics matter.
Second, compressing single-particle creation operators combines the preparation of a free fermionic or Hartree-Fock state into 
a subsequent
time evolution.
And finally, controlled evolution is an important yet difficult problem, and we show that it is possible to compress 
this type of circuits which have
a single control ancilla.
The most notable application is the compression of Hadamard-test style circuits, where wave function overlaps need to be measured. 
If a
controlled evolution with multiple ancilla is required, for example in the construction of a SELECT oracle, our single control compressed evolution circuit can be easily extended to admit multi-controlled evolution by introducing additional ancilla and Toffoli gates~\cite{Barenco}.

The paper is structured as follows. In \cref{sec:notation_and_recap}, we introduce the notations and the gates we will use throughout the paper, as well as the notions of blocks, block rules and compression theorems from Refs.~\cite{kokcu2022algebraic, camps2022algebraic}. We then introduce the block mappings from the same papers for TFIM and TFXY models, and show their connections to free fermions. 
In \cref{sec:swap_and_longrange}, we show that these TFIM and TFXY mappings can be used to compress free fermions on any lattice. We demonstrate this by simulating a $4 \times 4$ tight binding model on \textit{ibmq\_washington}. 
In \cref{sec:creation_annihilation}, we introduce a new block mapping which allows us to compress the creation-/annihilation of single fermions as well as fermion pairs. We demonstrate this by preparing the ground state of a TFXY model Hamiltonian by using adiabatic state preparation. In \cref{sec:qblock-qcompression}, we introduce a new mathematical structure ``\qblock", \qblock\ rules, and a $Q$-compression which allows us to compress both \qblock s and blocks together. 
In \cref{subsec:controlled_free_fermions} and \cref{subsec:controlled_free_fermions_creation}, we provide \qblock\ mappings that allow us to compress controlled evolutions of both free fermions, and free fermions with single fermion creation/annihilation, on any given lattice. We use the \qblock\ mapping for the free fermions to demonstrate the topological phase transition of the free Creutz-Hubbard model, by calculating ground state overlaps over different parameters in the phase space using Quantinuum's H1-1E emulator and H1-1 QPU hardware.

The compression software is available as part of the fast free fermion compiler (\texttt{F3C}) \cite{f3c, f3cpp} at \url{https://github.com/QuantumComputingLab}.
\texttt{F3C} is based on the \texttt{QCLAB} toolbox \cite{qclab,qclabpp} for creating and
representing quantum circuits.

\section{Notation and recap}\label{sec:notation_and_recap}

In this section, we will define short-hand notations for gates such as ``Pauli string rotations" and ``free fermionic gates", and re-state certain theorems and results from our previous works \cite{kokcu2022algebraic, camps2022algebraic} for completeness.

\subsection{Notation for circuit elements}

Throughout this paper, we will illustrate the following 1- and 2-qubit operators as below: 
\begin{align}
    &\includegraphics[height = 1.3in]{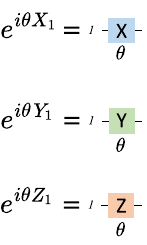}
    &\includegraphics[height = 1.3in]{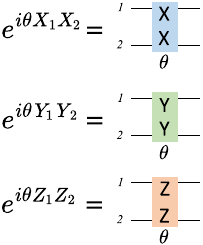}
\end{align}
In some figures, such as \cref{fig:all_block_things}(f) and (g), angles are omitted for convenience. Unless otherwise specified, it is assumed that the angle for each gate is a free parameter 
in these cases.
To extend the notation to generic Paui strings, we will frequently use gate representations of the forms $e^{i \theta X_1 Z_2 Z_3 \cdots Z_{n-1} X_n}$ and $e^{i \theta Y_1 Z_2 Z_3 \cdots Z_{n-1} Y_n}$. We illustrate these operators here for $n = 4$:
\begin{align*}
    \centering
    \includegraphics[width =.95\columnwidth]{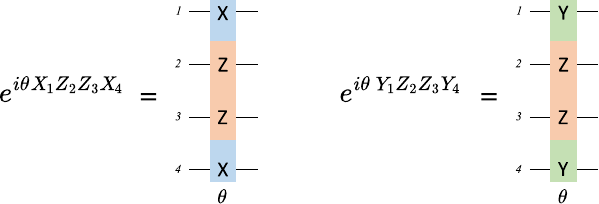}
\end{align*}
Note that these gates correspond to exponentials of single Pauli strings, which can be implemented via 1-qubit gates and CNOT gates \cite{nielsen2010quantum,kokcu2022fixed,magoulas2023cnot}.

\subsection{Block Definition}

In Refs.~\cite{kokcu2022algebraic,camps2022algebraic}, we introduced a novel mathematical operator referred to as a \emph{block} ($B$) and characterized by three specific properties. These defining features are succinctly outlined in Definition~\ref{def:block} below. A block is denoted by index $i$ and one or more parameters $\myvec{\theta}$, typically expressed as $B_i(\myvec{\theta})$. It can be associated with specific parameterized quantum gate operations. Because in this paper another type of block is defined, we refer to the block structure from~\cite{kokcu2022algebraic,camps2022algebraic} as \emph{\bblock}.

\begin{definition}[\bblock~\cite{kokcu2022algebraic,camps2022algebraic}]
\label{def:block}
Define a ``B-block'' $B_i(\myvec{\theta})$ as a structure that satisfies the following three properties:
\begin{enumerate}
  \item \textbf{Fusion:} for any set of parameters $\myvec{\alpha}$ and $\myvec{\beta}$, there exist $\myvec{a}$ such that
        \begin{equation}
          B_i(\myvec{\alpha}) \, B_i(\myvec{\beta}) = B_i(\myvec{a}),
        \end{equation}
  \item \textbf{Commutation:} for any set of parameters $\myvec{\alpha}$ and $\myvec{\beta}$ such that
        \begin{equation}
          B_i(\myvec{\alpha}) \, B_j(\myvec{\beta}) = B_j(\myvec{\beta}) \, B_i(\myvec{\alpha}), \qquad |i-j|>1,
        \end{equation}
  \item \textbf{Turnover:} for any set of parameters $\myvec{\alpha}$, $\myvec{\beta}$ and $\myvec{\gamma}$ there exist $\myvec{a}$, $\myvec{b}$ and $\myvec{c}$ such that 
        \begin{equation}
          B_i(\myvec{\alpha}) \, B_{i+1}(\myvec{\beta}) \, B_i(\myvec{\gamma}) = B_{i+1}(\myvec{a}) \, B_i(\myvec{b}) \, B_{i+1}(\myvec{c}).
        \end{equation}
\end{enumerate}
\end{definition}

These three properties are summarized in \cref{fig:all_block_things} panels (a-c).
Based on these properties, we have shown in~\cite{kokcu2022algebraic, camps2022algebraic} that any quantum circuit that only consist of \bblock s can be compressed to a triangle, also indicated in \cref{fig:all_block_things} panel (d),
where the block on the right follows the red line down, using repeated turnover operations, and fuses with the block at the end of the line.
Additionally, as shown in panel (e), a triangle can be transformed into a shallower structure, known as a square, using only B-block properties. 

\subsection{TFIM and TFXY Block Mappings}

\rredtext{In Refs.~\cite{kokcu2022algebraic,camps2022algebraic}, we introduced several $B$-block mappings which map a certain set of gates to a list of $B$-blocks, including mappings for the 1-D TFIM, TFXY, Kitaev and XY model.}
\rredtext{In this paper, we will 
only
be focusing on the TFIM and TFXY block mappings, 
which are}
shown in panels (f--g) of \cref{fig:all_block_things}. Mathematically, the TFIM block mapping is given by
\begin{align}\label{eq:tfim_block_mapping}
\begin{split}
    B^\TFIM_{2i-1}(\theta) &= e^{i \theta Z_{i}},\\
    B^\TFIM_{2i}(\theta) &= e^{i \theta X_{i} X_{i+1}},
\end{split}
\end{align}
where $i > 0$, and the TFXY block mapping is given by
\begin{align}\label{eq:tfxy_block_mapping}
\begin{split}
    B^\TFXY_{i}(\myvec{\theta}) = \ & e^{i \theta_1\: Z_i}e^{i \theta_2\: Z_{i+1}}e^{i \theta_3\: X_i X_{i+1}}\\
    &e^{i \theta_4\: Y_i Y_{i+1}}\:e^{i \theta_5\: Z_i}e^{i \theta_6\: Z_{i+1}}.  
\end{split}
\end{align}
\rredtext{The names TFIM and TFXY mapping come from the gate sets: the TFIM mapping includes only $XX$ and $Z$ gates, and naturally describes a 1-D TFIM Hamiltonian, and the TFXY mapping includes all $XX$, $YY$ and $Z$ gates which allows a 1-D TFXY model Hamiltonian to be naturally represented by the TFXY mapping. 
}
\redtext{Due to the frequent use of the TFXY mapping, we introduce the concept of \textit{free fermionic matchgates} to compactly represent these gates, defined as follows:}
\begin{align}\label{eq:matchgate_definition}
\vcenter{\hbox{\includegraphics[width = 0.7\columnwidth]{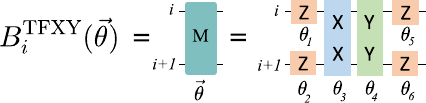}~}}.
\end{align}
\redtext{We label these gates with letter ``M" in the illustrations, as they are a specific instance of what is known in the literature as ``matchgates" \cite{bassman2022constant, projansky2024extending}. However, matchgates are defined differently in other works \cite{jozsa2008matchgates}, which are generalized versions of the gate above. 
We will refer to the gate defined in \cref{eq:matchgate_definition} as either a ``free fermionic matchgate" or a ``free fermionic gate", interchangeably, since they are equivalent to a combination of nearest-neighbor fermion hopping, pair creation-annihilation, and on-site chemical potentials under a Jordan-Wigner transformation.
}

It was shown in Ref.~\cite{kokcu2022algebraic} that the mappings~\eqref{eq:tfim_block_mapping} and \eqref{eq:tfxy_block_mapping} obey the \bblock\ properties. 
These two block mappings are equivalent to each other and can be used interchangeably: TFXY blocks can be transformed to TFIM blocks and vice versa.
For example, a conversion from TFIM to TFXY can be achieved as follows. Starting from a TFIM triangle on 2 qubits, we can regroup them into TFXY blocks with certain angles set to zero:
\begin{align}\label{eq:tfim2tfxy}
    \vcenter{\hbox{\includegraphics[width = 0.85\columnwidth]{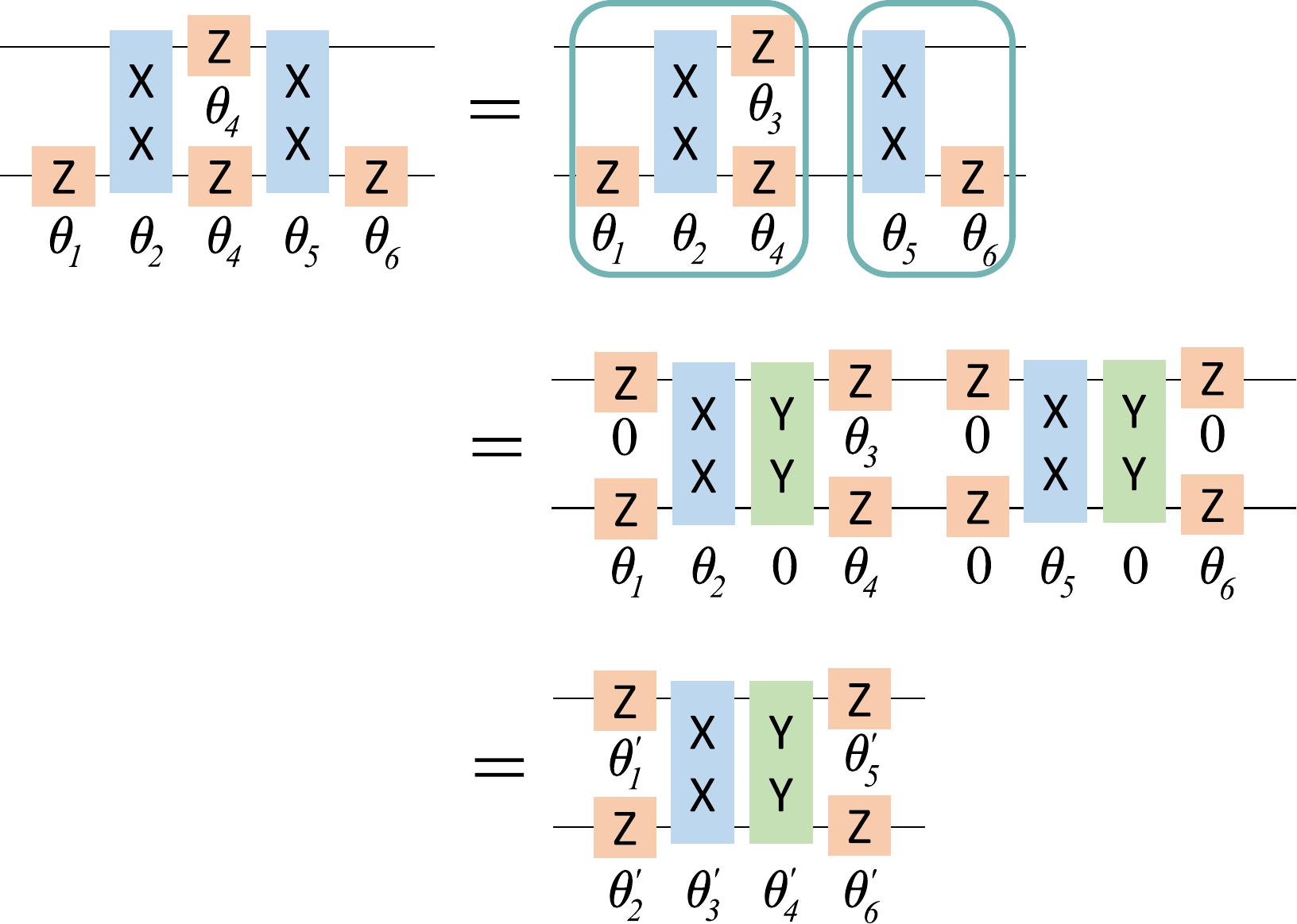}~}}
\end{align}
In the final step, we used the TFXY fusion operation given in \cite[Sec. III.D]{kokcu2022algebraic}. This TFIM $\rightarrow$ TFXY transformation will be used frequently in the following sections.
The reverse conversion can be realized in a similar way: a $YY$-rotation is just an $XX$-rotation with additional $Z$-rotations on both qubits, thus,
\begin{align}\label{eq:tfxy2tfim}
    \vcenter{\hbox{\includegraphics[width = 0.75\columnwidth]{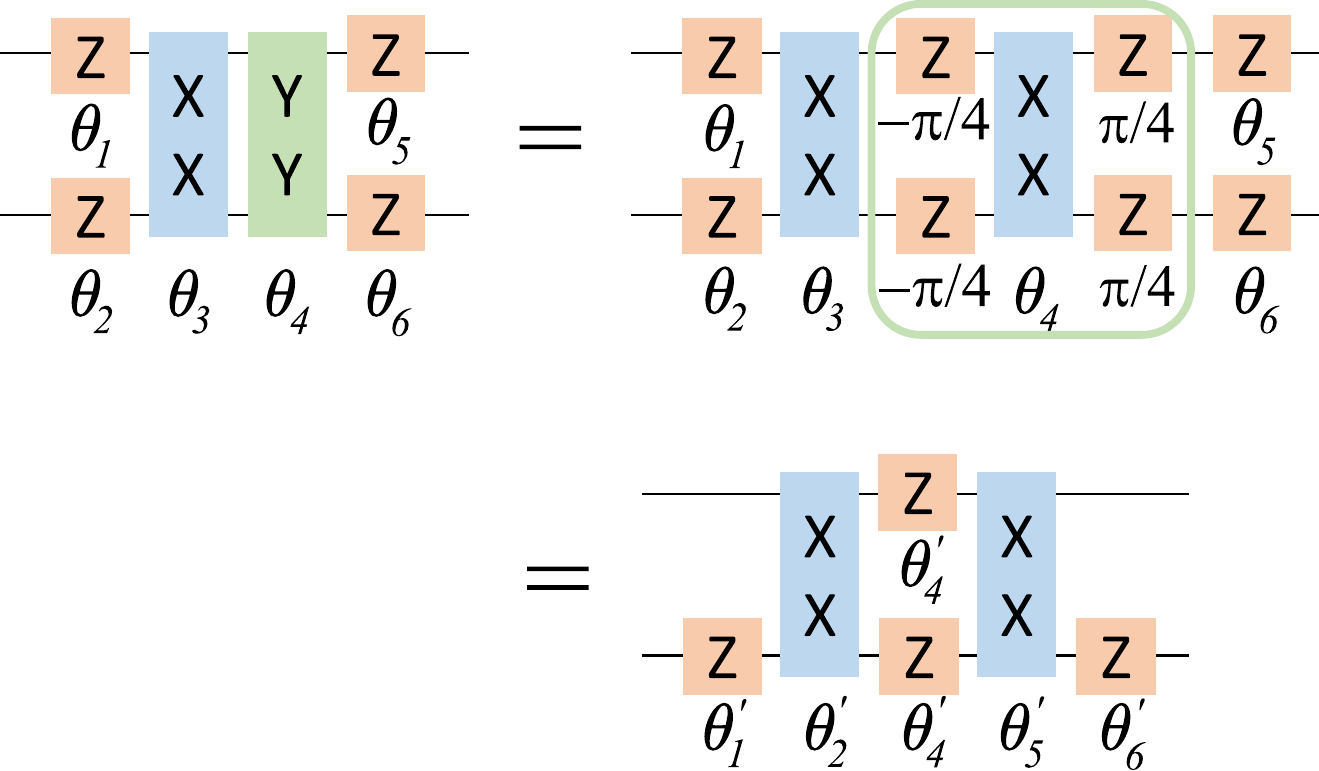}~}}
\end{align}
\rredtext{where in the last step we have used algebraic compression of TFIM blocks given in \cite[Sec. III.C]{kokcu2022algebraic} to rewrite the series of TFIM blocks as a TFIM triangle.}

The two mappings each have their own advantages. The TFXY mapping only requires 2 CNOTs per block~\cite{vidal2004universal} using the following circuit decomposition
\begin{align}\label{eq:matchgate}
    \vcenter{\hbox{\includegraphics[width = 0.75\columnwidth]{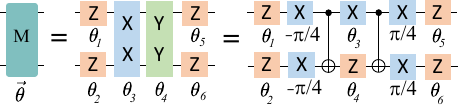}~}},
\end{align}
while the same structure in a TFIM mapping requires 4 CNOTs in general (see Eq.~\ref{eq:tfxy2tfim}), since each $XX$ rotation requires 2 CNOTs. The TFXY mapping thus only requires half as many CNOTs. On the other hand, compressing a TFIM mapping is approximately 20 times faster on a classical computer and more straightforward~\cite{camps2022algebraic}.
Throughout the paper, \rredtext{we will be using both of these mappings extensively, and benefit from the simplicity of TFIM in proofs, and CNOT efficiency of TFXY in the circuit constructions and hardware demonstrations.}

As shown in~\cite{kokcu2022algebraic}, the TFIM and TFXY models are connected to free fermions via the Jordan--Wigner transformation,
\begin{align}\label{eq:JW}
\begin{split}
    c_n &= \frac{1}{2} Z_1 Z_2 \cdots Z_{n-1} (X_n - i Y_n), \\
    c_n^\dagger &= \frac{1}{2} Z_1 Z_2 \cdots Z_{n-1} (X_n + i Y_n),
\end{split}
\end{align}
where $c_n$ (resp. $c_n^{\dagger}$) are the usual fermionic annihilation (resp. creation) operators.
Then, for example, we have
\begin{align}
\begin{split}
    c^\dagger_{n} c_{n+1} + \hc =& -\frac{1}{2}\left( X_n X_{n+1} + Y_n Y_{n+1}  \right), \\
    c^\dagger_{n} c_{n} =& -\frac{1}{2}\left( I - Z_n \right),
\end{split}
\end{align}
which can directly be represented by TFIM and TFXY block mappings.
In a similar way, $X_n X_{n+1}$, $Y_n X_{n+1}$, $X_n Y_{n+1}$, and $Y_n Y_{n+1}$ can be written as a linear combination of $c^{(\dagger)}_{n} c^{(\dagger)}_{n+1}$ terms.
With this, we have shown that the TFIM and TFXY block mappings can be used to compress time evolution circuit for any time dependent, free fermionic Hamiltonian on a 1D open chain with nearest-neighbor hopping $c_{i+1}^\dagger c_i$, nearest-neighbor pair creation (annihilation) $c_{i+1} c_i$ ($c_{i+1}^\dagger c_i^\dagger$), and onsite potential $c_{i}^\dagger c_i$ terms:
\begin{align}\label{eq:hamfree_only}
\begin{split}
    \ham(t) = \sum_{i = 1}^{n-1} \big( h_{i}(t) \, c_i^\dagger c_{i+1} + p_{i}(t) \, c_i c_{i+1} + u_i(t) c_i^\dagger c_i &\big) \\
    + \hc&
\end{split}
\end{align}

\section{New Block Mapping Results}
\label{sec:new_block_mappings}

In this section, we will show new results that we obtained with the already existing \bblock\ mappings. First, we will show that the TFXY and TFIM mappings are capable of representing long range free fermion hoppings as well, and apply this to a 2-D $4 \times 4$ tight binding model. Secondly, we show that with the addition of one more element, the TFIM mapping is capable of implementing fermion creation-annihilation as well. With the results of this section, we are now capable of compressing Hamiltonians that have quadratic and linear terms of $c_i$ and $c_i^\dagger$, which are of the following form:
\begin{align}\label{eq:lrfermions_c1}
   \ham(t) = \sum_{i,\mathbf{j}} \big( h_{i\mathbf{j}}(t) \, c_i^\dagger c_{\mathbf{j}} + p_{i\mathbf{j}}(t) \, c_i c_{\mathbf{j}} +
   \mathbf{q_i(t) \, c_i} \big) + \hc,
\end{align}
where long range and creation terms highlighted in bold are new.

\subsection{Fermionic Swap Gate and Long-Range Hoppings}
\label{sec:swap_and_longrange} 

First, we show that the TFXY and TFIM block mappings can compress any free fermionic Hamiltonian with long-range hoppings,
\begin{align}\label{eq:free_fermion_ham}
    \ham(t) = \sum_{i,\mathbf{j}} \big( h_{i\mathbf{j}}(t) \, c_i^\dagger c_\mathbf{j} + p_{i\mathbf{j}}(t) \, c_i c_\mathbf{j} \big) + \hc \, ,
\end{align}
where the novelty is the arbitrary lattice connectivity between any $i,j$ pair as highlighted by the boldface $\mathbf{j}$.
\redtext{
To illustrate these long range terms, we define the following generalization of the free fermionic matchgates given in \cref{eq:matchgate_definition}:
}
\begin{align}\label{eq:generalized_matchgate}
    \centering
    \includegraphics[width = 0.55\columnwidth]{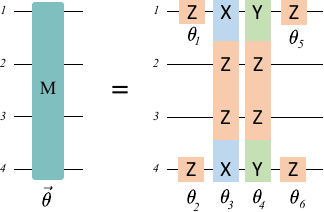}.
\end{align}
\redtext{ 
The above is an illustration of the gate between qubits 1 and 4, however it can be generalized to any pair of qubits.
These gates are equivalent to exponentials of linear combinations of Jordan-Wigner transformed versions of long range fermion hopping $(c_1^\dagger c_4, c_4^\dagger c_1)$, pair creation-annihilations $(c_1^\dagger c_4^\dagger, c_4 c_1)$, and on site chemical potentials $(c_1^\dagger c_1, c_4^\dagger c_4)$, i.e.,
\begin{equation}
    e^{i (z c_1^\dagger c_4 + z^* c_4^\dagger c_1 +w c_1^\dagger c_4^\dagger +  w^* c_4 c_1 + \alpha c_1^\dagger c_1 + \beta c^\dagger_4 c_4)},
\end{equation}
where $z,w \in \mathbb{C}$ complex numbers, and $\alpha, \beta \in \mathbb{R}$ real numbers. The $Z$ chain between the end qubits 1 and 4 results from the $Z$ chain given in the Jordan-Wigner transformation given in \cref{eq:JW}.
}

\begin{figure*}[htpb]
    \centering
    \begin{tikzpicture}
    \node[anchor = north, inner sep = 0](image) at (0,0){
    \includegraphics[width=1.8\columnwidth]{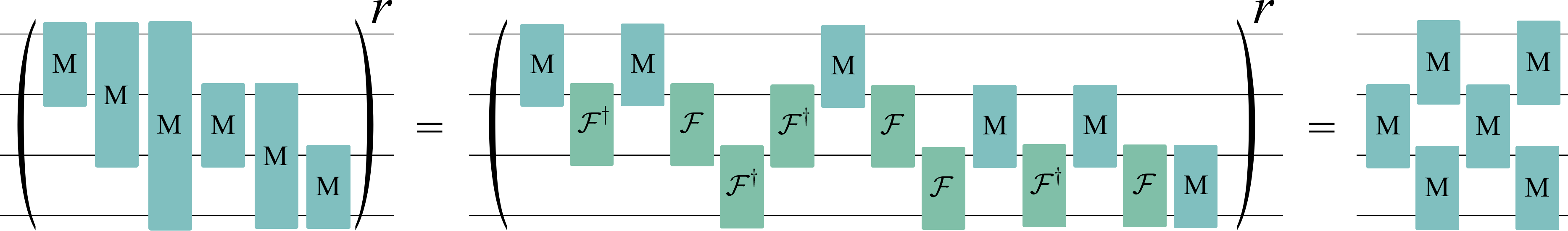}};
    \node[anchor = north, inner sep = 0](image) at (0,-2.5){
    \includegraphics[width = 2\columnwidth]{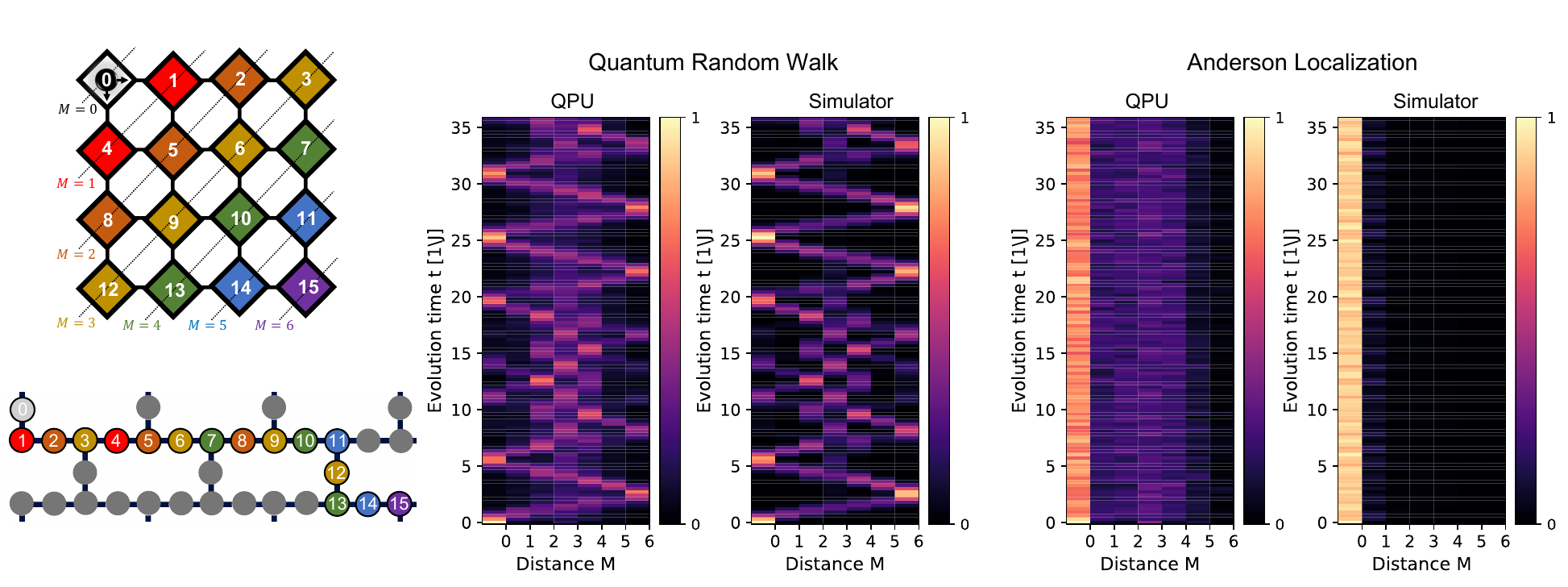}};
    \node[black, fill = white] at (-8.5, 0){(a)};
    \node[black, fill = white] at (-8.5, -2.9){(b)};
    \node[black, fill = white] at (-8.5, -6.4){(c)};
    \node[black, fill = white] at (-3.5, -2.9){(d)};
    \node[black, fill = white] at (2.9, -2.9){(e)};
    \end{tikzpicture}
    \caption{\rredtext{(a) An illustration of compression of long range hoppings. Any non-local fermion hopping can be written in terms of nearest neighbor hoppings as shown in \cref{eq:long_range_hopping}, which then can be compressed into a TFXY square. }
    The other panels illustrate the simulation results from \emph{ibm\_washington} of a free-fermion on a two-dimensional (2D) lattice. Specifically:
    (b) Schematic of the 16-site, 2D lattice, color-coded by the distance $M$ each lattice site is from the reference lattice site, labeled `0'.  The fermion is initialized at reference lattice site `0', and is allowed to evolve freely in time. (c) The topology of the qubits in the \emph{ibmq\_washington} quantum processing unit (QPU) is shown, along with how the lattice sites from the schematic in panel (b) are mapped to the particular 16 qubits used in these simulations. (d) The occupation number at each distance $M$ versus time with no disorder in the lattice.  This leads to ballistic transport of the fermion.  The left-hand plot shows results from the QPU, while the right-hand plot shows results from a noise-free quantum simulator. (e) The occupation number at each distance $M$ versus time with random disorder in the lattice.  This leads to Anderson localization of the fermion.  The left-hand plot shows results from the QPU, while the right-hand plot shows results from a noise-free quantum simulator.
    }
    \label{fig:2DFF_simulation}
\end{figure*}

\redtext{In order to show that time evolution under the long range free fermionic Hamiltonian given in \cref{eq:long_range}}, we will use the fermionic swap (FSWAP) gate, which is defined as follows
\begin{align}\label{eq:fswap_matrix}
    \fswap_{j,j+1} =  i \begin{bmatrix} 
                        1 & 0 & 0 & 0\\
                        0 & 0 & 1 & 0\\
                        0 & 1 & 0 & 0\\
                        0 & 0 & 0 & -1
                    \end{bmatrix},
\end{align}
for nearest neighbor qubits $j$ and $j+1$.
Up to a global phase $i$, this operation is equivalent to a swap operation that keeps track of the sign generated by fermion exchange, which leads to the following relations
\begin{align}\label{eq:fermion_shift}
\begin{split}
    \fswap_{i,i+1} \ c_i^{(\dagger)} \ \fswap_{i,i+1}^\dagger &= c_{i+1}^{(\dagger)}, \\ \fswap_{i,i+1} \ c_{i+1}^{(\dagger)} \ \fswap_{i,i+1}^\dagger &= c_i^{(\dagger)}.
\end{split}
\end{align}
With the global phase $i$, $\fswap_{j,j+1}$ can be mapped to a TFXY block~\eqref{eq:tfxy_block_mapping} with index $j$, and parameters $\theta_1 = \theta_2=0$ and $\theta_3=\theta_4=\theta_5=\theta_6=\pi/4$. 
It follows that the $\fswap_{j,j+1}$ operation can thus be compressed within the existing TFXY \bblock\ compression algorithm.

From Eq. \eqref{eq:fermion_shift}, we can use FSWAP gates to generate long-range quadratic terms from nearest-neighbor terms,
\begin{align}\label{eq:long_range}
    c_1^{(\dagger)} c_4^{(\dagger)} = \fswap_{3,4}\fswap_{2,3}  \ c_1^{(\dagger)} c_2^{(\dagger)} \ \fswap_{2,3}^\dagger\fswap_{3,4}^\dagger.
\end{align}
\redtext{This expression is an exponent of a long range free fermionic gate, thus we can illustrate \cref{eq:long_range} as follows}
\begin{align}\label{eq:long_range_hopping}
    \centering
    \vcenter{\hbox{\includegraphics[width = .75 \columnwidth]{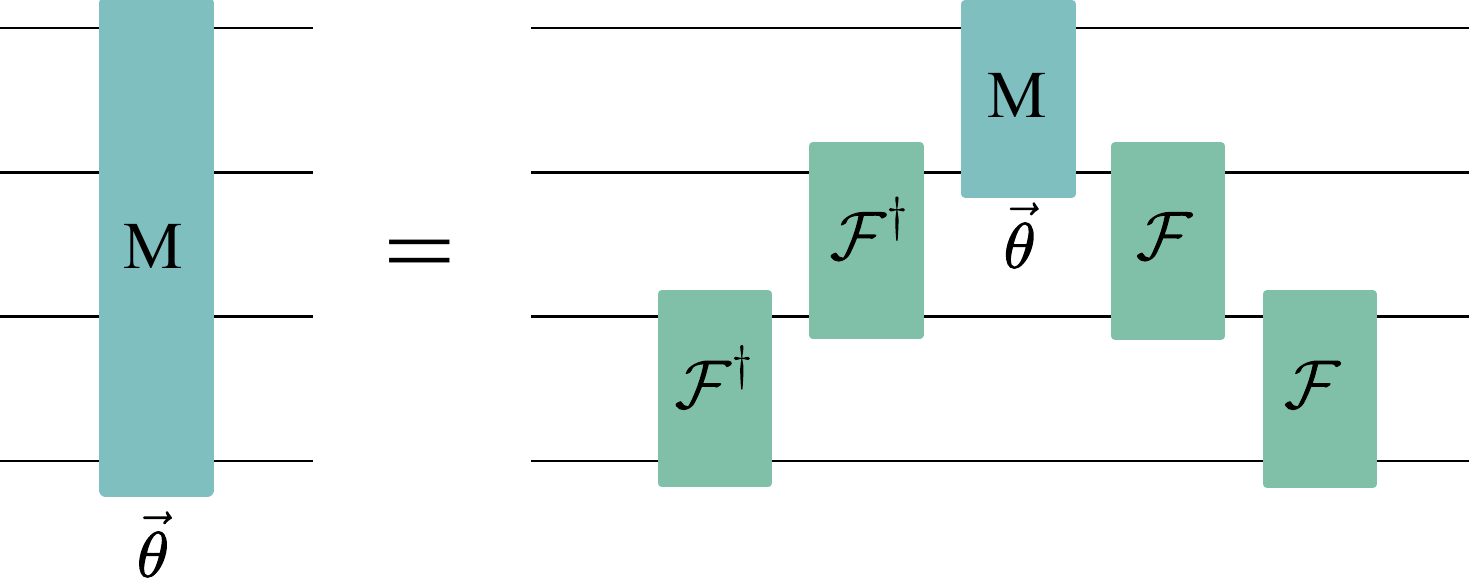}}}.
\end{align}
\rredtext{Here we illustrated fermionic swap (dagger) gates as free fermionic gates with label $\mathcal{F}^{(\dagger)}$, 
\redtext{and with a different color to specify them.}
This way we can decompose a long range hopping into 
\redtext{a series of nearest neighbor free fermionic matchgates, i.e.}
TFXY \bblock s.}
This means that the TFIM and TFXY block mappings given in \cref{eq:tfim_block_mapping,eq:tfxy_block_mapping} can generate long-range quadratic terms, and thus can be used to compress quadratic fermionic Hamiltonians on any lattice. 
An illustration of this approach is shown in \cref{fig:2DFF_simulation}(a),
where we schematically decompose and compress 
$r$ long-range fermionic hoppings.

\rredtext{We utilize our long range free fermion compression to}
simulate transport of free-fermions 
\rredtext{on a 2D lattice}
on \textit{ibmq\_washington}
\cite{karamlou2022quantum}.  In particular, we consider a fermion initialized on a reference lattice site, and observe how it evolves.
The lattice, pictured in \cref{fig:2DFF_simulation}(b), comprises 16 sites with 
open
boundary conditions. The lattice site labeled `0' is considered our reference site, where the fermion is initialized.  The lattice sites are color-coded by their taxicab or Manhattan distance $M$ to the reference site.  Each lattice site is mapped to a qubit on the \emph{ibmq\_washington} quantum processing unit (QPU), which is 
an IBM Quantum
Eagle r1 processor.
Note that the 2D lattice of sites is mapped to a one-dimensional chain of qubits on the QPU, and thus requires long-range hopping operations as in~\cref{eq:long_range_hopping}.  For example, in the 2D lattice, the fermion can hop between nearest-neighbor sites 11 and 15 (see \cref{fig:2DFF_simulation}(b)), but when the lattice is mapped to the qubits in the QPU, sites 11 and 15 are no longer neighbors (see \cref{fig:2DFF_simulation}(c)), and thus a long-range interaction between these qubits must be implemented.  

We examine transport of a fermion on this 16-site lattice both with and without disorder.  To do this, we track the occupation number at varying distances $M$ from the reference site as the fermion evolves freely through time.  When there is no disorder in the system, we expect the fermion to behave ballistically and oscillate back and forth within the lattice. \cref{fig:2DFF_simulation}(d) shows results from simulating a free fermion on a 2D lattice with no disorder on the \emph{ibmq\_washington} QPU as well as on a noise-free quantum simulator.  The colors in the plots correspond to the occupation number at each distance $M$ versus evolution time.  When there is large random disorder in the system, we expect the fermion to exhibit Anderson localization \cite{andersonloc1, andersonloc2}. \cref{fig:2DFF_simulation}(e) shows results from simulating a free fermion on a 2D lattice with large random disorder on the \emph{ibmq\_washington} QPU as well as on a noise-free quantum simulator.  Simulations on the QPU were performed with 50,000 shots and any shot that did not conserve particle number was discarded.  Two straightforward error mitigation techniques were also used to reduce noise in the results from the QPU.  The first was a scalable readout error mitigation method implemented with the $mthree$ package \cite{nation2021scalable}, which reduces errors in quantum measurement via calibration.  The second was dynamical decoupling \cite{viola1998dynamical, zanardi1999symmetrizing, vitali1999using, duan1999suppressing}, a method that can suppress qubit decoherence via the application of a set of pulses (which together amount to application of the identity operator) to idling qubits which cancels the system-environment interaction \cite{pokharel2018demonstration}.  

While the results from the QPU still exhibit some level of noise, there is good qualitative agreement with the results from the exact simulator, and a difference in transport can clearly be identified between lattices with and without disorder.  We attribute the impressive results of these 16-qubit dynamic simulations on a real QPU to our compression algorithm, which produces short-depth circuits that do not grow with increasing evolution time.

\subsection{Fermion Annihilation and Creation}
\label{sec:creation_annihilation}

\begin{figure*}[htpb]
    \centering
    \includegraphics[width = 1.5 \columnwidth]{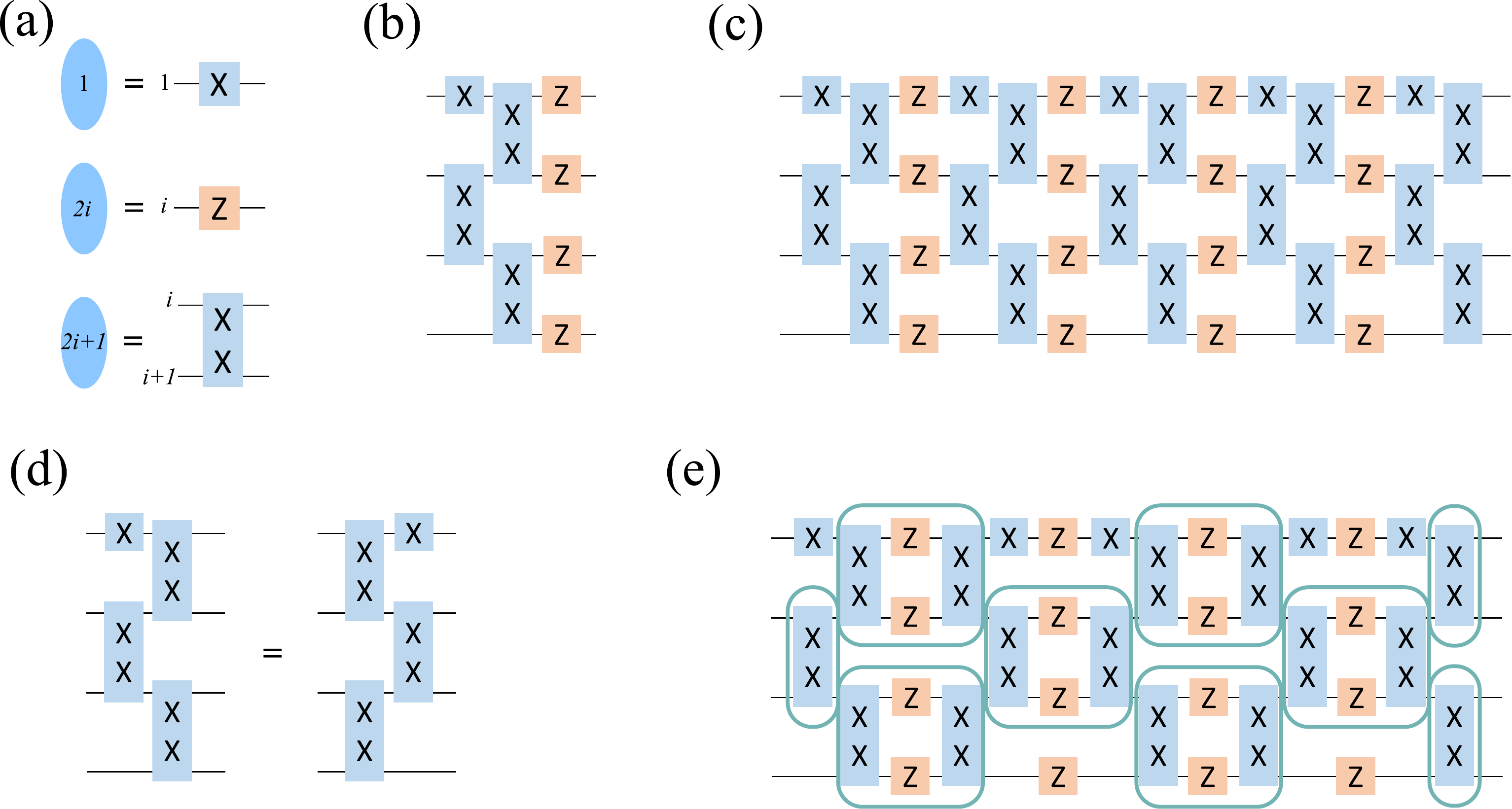}
    \caption{In panel (a), we show the block mapping that covers the free fermions with creation. By using these, we can build fermionic swaps, and carry the creation-annihilation operator on the first qubit to any other qubit. Panel (b) show the complete set of blocks for 4 qubits. In panel (c), we show the final circuit representing the square structure with this particular block mapping, that can be obtained by using the compression theorems. In this form, the circuit requires $2(n-1)(n+1) = 2(n^2-1) = O(2n^2)$ CNOT gates. By using the relation given in panel (d) on every other $X-XX$ vertical stripe, certain $XX$ gates can be grouped as shown in panel (e). Then using the relation \cref{eq:tfim2tfxy}, these groups can be transformed into free fermionic gates. After this simplification, the number CNOTs is reduced to $(n-1)(n+2) = O(n^2)$, which is approximately half of the CNOT count of the circuit in panel (c).}
    \label{fig:cTFIM_blocks}
\end{figure*}

In this section, we demonstrate that the preparation of a product state of any number of free fermionic states can be compressed. 
Consider the following Hamiltonian:
\begin{align}\label{eq:hamannihilation}
    \ham(t) =& \sum_{i,j} \big( h_{ij}(t) \, c_i^\dagger c_{j} + p_{ij}(t) \, c_i c_{j} \big) + \hc \nonumber \\
    &+ \sum_i \big( \mathbf{q_i(t) \, c_i +  q_i(t)^* c_i^\dagger} \big),
\end{align}
where $h_{ij}(t)$, $p_{ij}(t)$ and $q_i(t)$ are complex functions of time.
The difference between this Hamiltonian and 
\cref{eq:free_fermion_ham} is the presence of the lone operators $c_i$ and $c_i^\dagger$, which are highlighted in bold, and annihilate (resp. create) fermions on site $i$. These terms cannot be represented by the standard TFIM and TFXY mappings from \cref{eq:tfim_block_mapping,eq:tfxy_block_mapping}. 

The evolution under the Hamiltonian given in \cref{eq:hamannihilation} yields full control over the particle content of a fermionic state, and allows us to generate any fermionic Gaussian state, i.e. fermion product state \cite{surace2022fermionic,collura2024quantum}. The pair creation terms $c_i c_j$ allow us to change the particle number by 2, and a repeated application of them leads to a change of particle count by an even number. 
The addition of $c_i$ terms allows us to change the particle content by an odd number as well, giving us full control. In fact, using the quadratic terms $c_i^{(\dagger)} c_j^{(\dagger)}$ as basis change and Bogolyubov transformations, by $c_i^{(\dagger)}$, we can change the particle content in any given basis, and add/subtract as many particles as desired in a unitary fashion. In \cref{asec:creating_fermions}, we show how to use the evolution under the Hamiltonian \cref{eq:hamannihilation} to add/subtract fermions in the momentum basis.

To compress the time evolution of \cref{eq:hamannihilation}, we extend the TFIM mapping as follows: 
\rredtext{\begin{align}\label{eq:creation_tfim_mapping}
\begin{split}
    B^\CTFIM_1(\theta) &= e^{i \theta X_1} \\
    B^\CTFIM_{2i}(\theta) &= e^{i \theta Z_{i}},\\
    B^\CTFIM_{2i+1}(\theta) &= e^{i \theta X_{i} X_{i+1}},
\end{split}
\end{align}
where the superscript $\CTFIM$ is short for creation TFIM, indicating that this mapping contains creation and annihilation of free fermions, and implements free fermions via TFIM terms.
An illustration of this mapping is given in Fig. \ref{fig:cTFIM_blocks}(a).
Note that this mapping is the same as the TFIM block mapping given in \cref{eq:tfim_block_mapping} with the addition of $B^\CTFIM_1$ and a shift of block indices, i.e. $B^\CTFIM_{i+1} = B^\TFIM_i$ for $i=1,2,3,\dots$. We will show that this mapping covers all interaction terms in \cref{eq:hamannihilation}, and that they satisfy the \bblock\ rules given in \cref{def:block}.}

\rredtext{The set of gates given in \cref{eq:creation_tfim_mapping} enables us to implement all the terms of the Hamiltonian given in \cref{eq:hamannihilation}. The operator in the exponent of the new block $B_1^{\CTFIM}$ is $X_1 \equiv c_1 + c_1^\dagger$. Together with $B^\CTFIM_2(\theta) = e^{i \theta Z_1}$, this new term enables this mapping to cover any Hermitian linear combination of $c_1$ and $c_1^\dagger$ via the following}
\begin{align}
\begin{split}
B^\CTFIM_2(\theta) \:  X_1& \: B^\CTFIM_2(-\theta) \\ =& \cos 2\theta \: X_1 + \sin 2\theta \: Y_1  \\ 
=& e^{2i\theta} c_1 + e^{-2i\theta} c_1^\dagger.
\end{split}
\end{align}
\rredtext{Thus, the $c_1^{(\dagger)}$ term in the Hamiltonian \eqref{eq:hamannihilation} for any complex function of time $q_1(t)$ can be written via the mapping given in \cref{eq:creation_tfim_mapping}. Using the fermionic swap gates, which are part of the TFIM mapping as shown previously, we can move the $c_1^{(\dagger)}$ term around to generate any $c_i^{(\dagger)}$ term for any function $q_i(t)$ by using \cref{eq:fermion_shift}.}
For example, $c^{(\dagger)}_4$ can be written in terms of $c^{(\dagger)}_1$ and fermionic swap operators as,
\begin{align}\label{eq:fermion_carry}
    c^{(\dagger)}_4 = \fswap_{3,4}\fswap_{2,3}\fswap_{1,2} \ c^{(\dagger)}_1 \ \fswap_{1,2}^\dagger\fswap_{2,3}^\dagger\fswap_{3,4}^\dagger,
\end{align}
which we illustrate as,
\begin{align}
    \includegraphics[width = .75 \columnwidth]{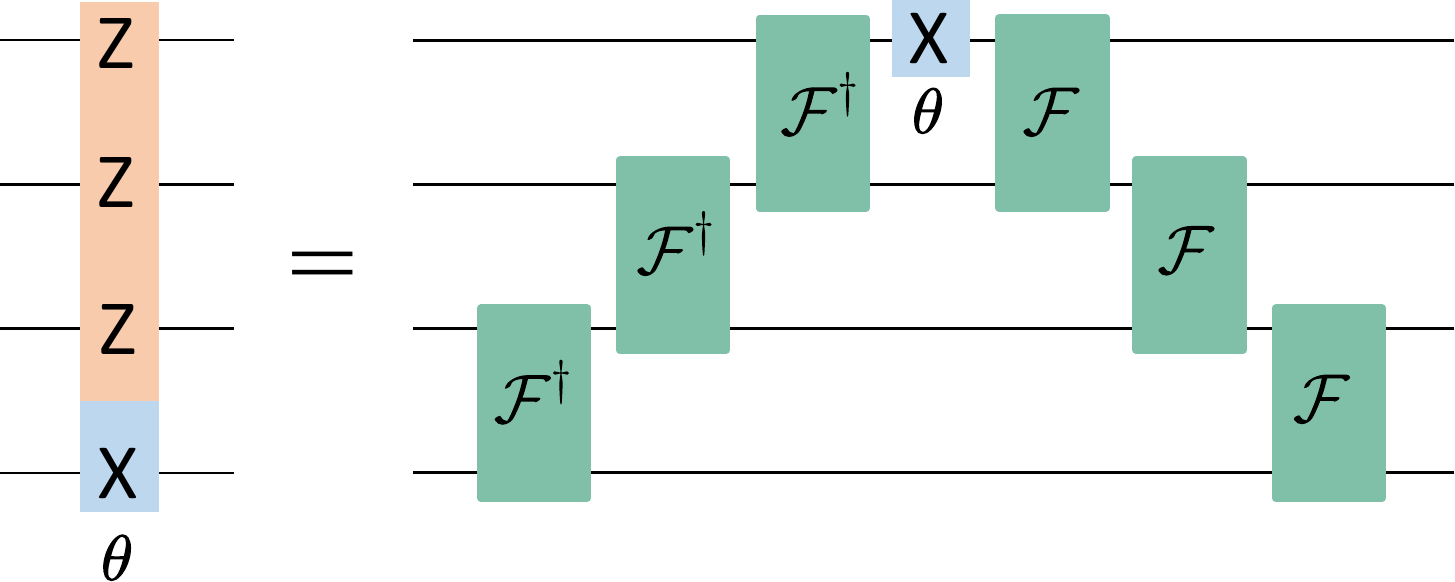}
\end{align}
On the left hand side representing $c^{(\dagger)}_4$, we have a rotation generated via Pauli $X_4$ with a $Z$-tail due to the Jordan-Wigner transformation~\eqref{eq:JW}.

\rredtext{
The mapping in \cref{eq:creation_tfim_mapping} satisfies all \bblock\ rules given in \cref{def:block}. We already know that $\{B^\CTFIM_i, i \geq 2 \}$ is a \bblock \: mapping, thus we only need to check if the $B^\CTFIM_1$ satisfies fusion, commutation and turnover rules. Fusion is satisfied since we have
\begin{align}
    B^\CTFIM_1(\alpha) B^\TFIM_1(\beta) = B^\CTFIM_1(\alpha + \beta).
\end{align}
Commutation with $B^\CTFIM_i$ for $i > 3$ is trivially satisfied since they do not share any qubits, and it commutes with $B^\CTFIM_3$ since $X_1$ commutes with $X_1 X_2$ even though they both act on qubit 1. Finally, the turnover property is satisfied directly because of the following: the operators in the exponents of $B_1^\CTFIM$ and $B_2^\CTFIM$, i.e., $X_1$ and $Y_1$, form the following representation of $\mathfrak{su}(2)$:
\begin{align}
    \mathfrak{su}(2) \equiv i \mathrm{span} \{ X_1, Z_1, Y_1 \}.
\end{align}
The Euler decomposition of this represantation yields that there exist $a,b,c \in \mathbb{R}$ for any $\alpha, \beta, \gamma \in \mathbb{R}$ such that 
\begin{align}
    e^{i a X_1} e^{i b Z_1} e^{i c X_1} = e^{i \alpha Z_1} e^{i \beta X_1} e^{i \gamma Z_1}, 
\end{align}
which is equivalent to
\begin{align}
\begin{split}
    B^\CTFIM_1(a) &B^\CTFIM_2(b) B^\CTFIM_1(c) \\ = &B^\CTFIM_2(\alpha) B^\CTFIM_1(\beta) B^\CTFIM_2(\gamma), 
\end{split}
\end{align}
where the corresponding angles can be calculated via \cite[Eqs. (29) and (30)]{kokcu2022algebraic}.
Thus, $B^\CTFIM_1$ satisfy the turnover property as well, and the mapping \cref{eq:creation_tfim_mapping} is a \bblock \: mapping.
}

This new mapping enables compression of
single particle creation/annihilation \rredtext{together with pair creation/annihilation and fermion hoppings. Using the compression theorem for \bblock s, the fixed depth circuits for evolution under the Hamiltonian \cref{eq:hamannihilation} are given in \cref{fig:cTFIM_blocks}. Panel (a) illustrates the CTFIM mapping given in \cref{eq:creation_tfim_mapping}, where we have blocks on the left hand side, and gates on the right hand side. In panel (b), we illustrate a full list of blocks for $n = 4$ qubits. As it can be seen, there are $2n = 8$ blocks. In panel (c), we show the circuit obtained from the square structure (see \cite{kokcu2022algebraic, camps2022algebraic} and \cref{fig:all_block_things}(e)) by using the block mapping given in panel (a). In this form, the circuit contains $(n+1)(n-1) = n^2-1$ $XX$ gates, which yields $2n^2-2$ CNOTs. This CNOT count can be reduced by moving $X_1$ and $X_i X_{i+1}$ gates around as shown in panel (d). This way, certain $XX$ gates can be brought and grouped together as shown in panel (e).
These groups then form TFIM triangles, and can be transformed into a TFXY block or a free fermionic gate by \cref{eq:tfim2tfxy}.
In this way, each group will cost only 2 CNOTs rather than 4. In total, the final CNOT count of the circuit given in panel (e) tuns out to be $(n-1)(n+2)$, which is approximately half of the CNOT count of the circuit in panel (c).
}

We will use an adiabatic state preparation example to demonstrate the compression of free fermions with creation and annihilation.
We will be generating the ground state of the following Hamiltonian
\begin{align}\label{eq:tight_binding}
\mathcal{H}_0 = -\tilde{t} \sum_{i=1}^{n-1}   \left(c^\dagger_i c^{\phantom{\dagger}}_{i+1} + c^\dagger_{i+1} c^{\phantom{\dagger}}_{i} \right) -\mu \sum_{i=1}^{n}   c^\dagger_i c^{\phantom{\dagger}}_{i},
\end{align}
for $\tilde{t} = 1$ and $\mu = 0$. For large negative values of $\mu$, the chemical potential term will dominate and the ground state will be the empty state $\ket{000...0}$, which is easy to prepare. Thus, if we change $\mu$ sufficiently slowly, by the adiabatic theorem we can prepare the ground state of the Hamiltonian with $\mu=0$ as long as the energy gap between the ground state and the first excited state does not become zero. This method is called adiabatic state preparation (ASP).

ASP comes with competing challenges for the current noisy hardware. 
First, the Hamiltonian needs to be changed adiabatically ---  requiring a small $d\mu/dt$ and a long evolution time $T$.
At the same time, due to the Trotter decomposition,
the time step $\delta t$ should be small. Altogether, this leads to a large number of Trotter steps $r = T / \delta t \gg 1$. For the non-compressable cases, it is difficult to simulate $r \gg 1$ due to hardware noise.  

A second challenge to adiabatically prepare the ground state of $\ham_0$ stems from the symmetries of the Hamiltonian. 
The Hamiltonian \eqref{eq:tight_binding} conserves particle number,
which makes the adiabatic state preparation approach difficult.
When $\mu = 0$ the Hamiltonian \eqref{eq:tight_binding} is particle-hole symmetric, and the ground state will be half filled. However, the initial state has no particles in it. Thus, without a symmetry breaking term, it is not possible to reach the ground state of $\ham_0(\mu = 0)$ since evolution under \cref{eq:tight_binding} does not change the particle number. 
More generally, the system has protected level crossings in the spectrum which makes the adiabatic state preparation impossible, unless an additional symmetry breaking term is added to the Hamiltonian to open gaps at those level crossing points.

In the TFXY model (which is equivalent to
\cref{eq:tight_binding} after a Jordan-Wigner transformation),
a global transverse $X$-direction magnetic field can generate
gaps in the spectrum \cite{francis2022adiabatic}.
A global magnetic field is however not compressible; a field on the second site ($X_2 = Z_1 \big(c_2 + c^\dagger_2 \big)$) is cubic in fermion terms, which makes the dynamical Lie algebra grow exponentially with the system size.
However, as we have shown, a field on the first site ($c_1 + c^\dagger_1 = X_1$) is compressible. Thus, we will add a single term that  creates/annihilates particles on site 1,
\begin{align}
    \ham'(t) = \lambda(t) \big( c_1 + c^\dagger_1 \big),
\end{align}
which in spin language corresponds to a local magnetic field in $X$ direction. This is sufficient to generate gaps between the ground state and the first excited state. Thus, the addition of $\ham'$ addresses both challenges listed above: it opens up gaps by breaking the particle conservation symmetry, and its evolution together with $\ham_0(\mu)$ is compressible into a circuit that is independent of the number of Trotter steps $r$ which needs to be large for ASP.

\begin{figure}[t]
    \centering
    \includegraphics[clip=true,trim=0 10 0 0,width = 0.99\columnwidth]{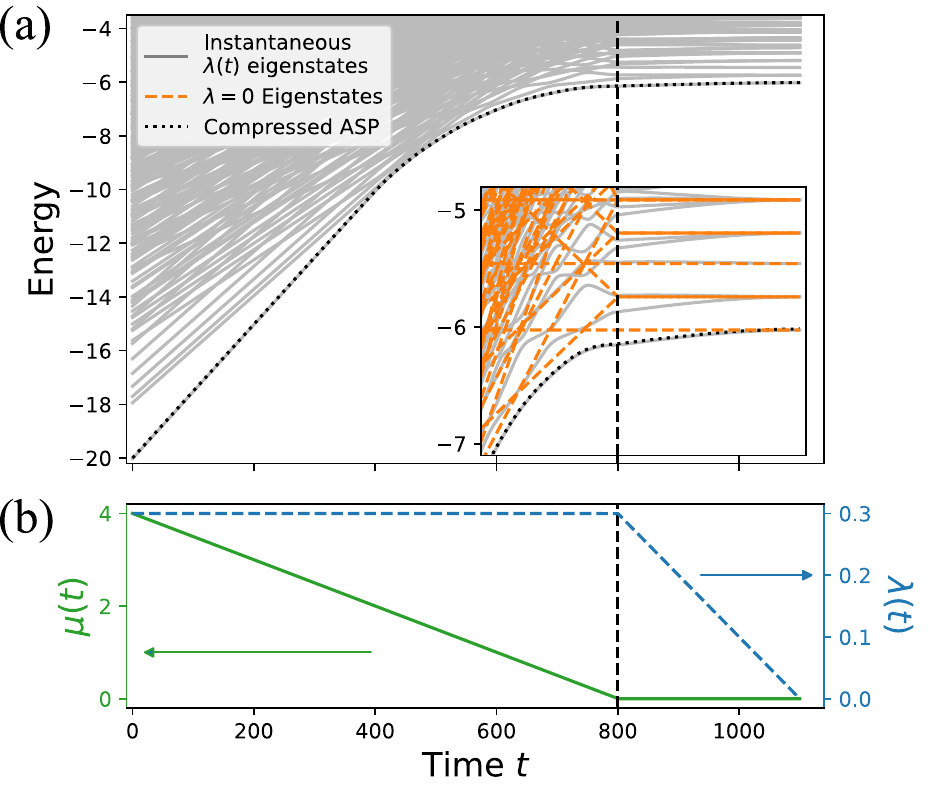}
    \caption{Numerical results for the adiabatic state preparation for the 1D chain via fermionic compression.
    (a) Instantaneous eigenstates and the result of the compressed time evolution with the parameters discussed in the text. Inset: a close-up view near the end of the evolution, with the instantaneous eigenvalues of $\ham_0$ (orange dashed lines) and $\ham=\ham_0 + \ham'$ (gray lines).
    (b) Time evolution of the chemical potential $\mu(t)$ and symmetry-breaking field $\lambda(t)$.}
    \label{fig:ASP_results}
\end{figure}

\begin{figure*}[htpb]
    \centering
    \includegraphics[width = 1.5 \columnwidth]{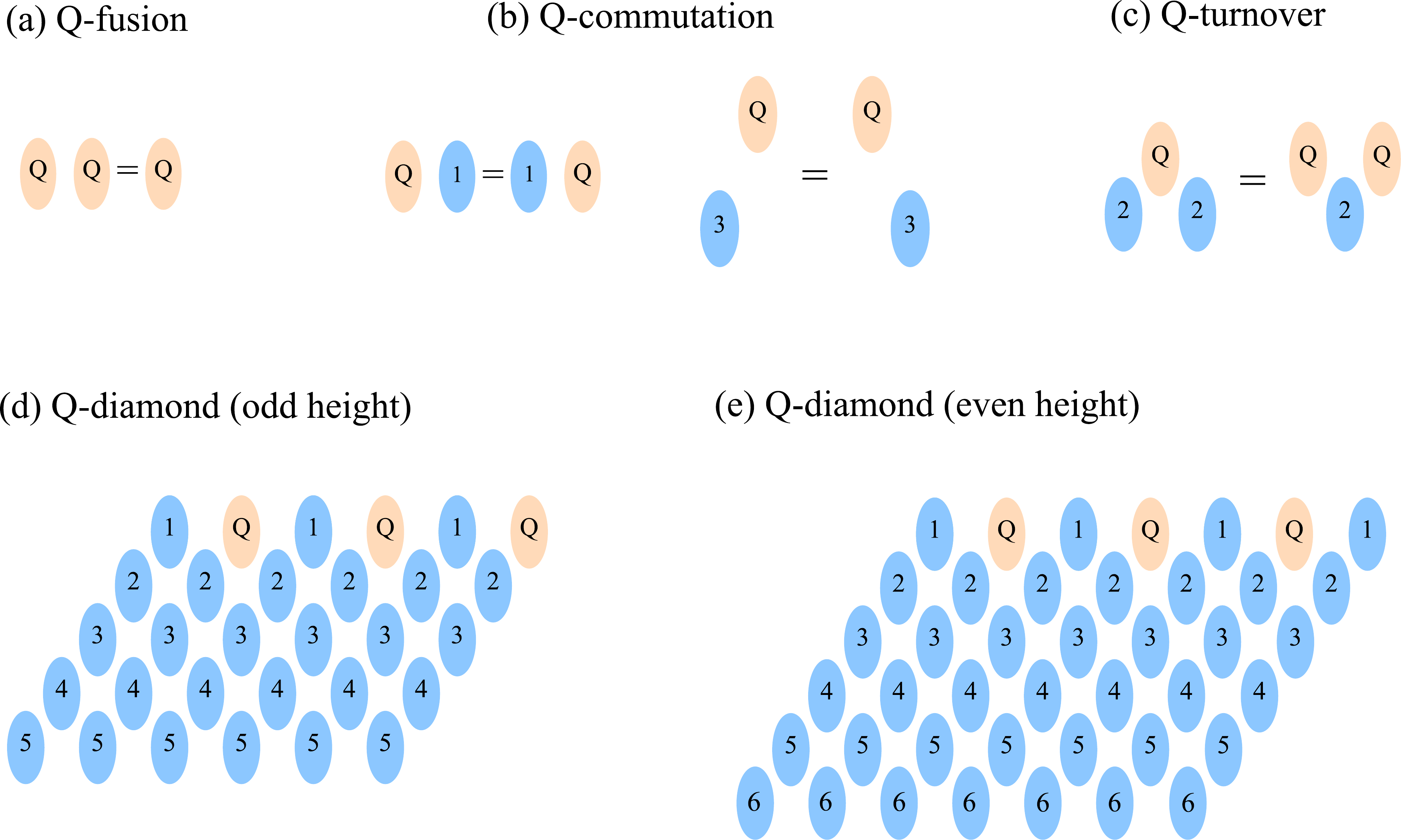}
    \caption{(a-c) $Q$-block properties given in \cref{def:q-block}. (d-e) The diamond structure defined in \cref{def:q-diamond} with heights $n=5$ and $n=6$. Notice that the alternating nature of $B_1$ and $Q$ in the diamond makes the top right corner of the odd height and even height diamonds differ from each other. This leads to different compression sequences for even-odd height diamonds.
    }
    \label{fig:all_qblock_things}
\end{figure*}

\cref{fig:ASP_results} shows the numerical results of adiabatic state preparation of the ground state of $\ham_0$ for $n=10$ sites via state vector simulation.
We initiate our state with no particles $\ket{\psi(t = 0)} =\ket{000...0}$ with $\lambda(t=0) = 0.3$ and $\mu(t = 0) = -4$. The amplitude of the chemical potential is large enough to ensure that the $0-$particle state has large overlap with the ground state, even with non-zero $\lambda$. We evolve the state with $d\mu/dt = -0.005$ and constant $\lambda = 0.3$ to keep the gaps open until we reach $\mu(t) = 0$ at $t=800$. As a final step, we slowly turn off $\lambda$ with $d\lambda/dt = -0.001$ and $\mu = 0$ until we reach $\lambda = 0$ to obtain the true ground state of $\ham_0$ in \cref{eq:tight_binding} at $\mu=0$. 

As can be seen in \cref{fig:ASP_results}, the energies obtained through adiabatic evolution compare well to the instantaneous ground state energies of $\ham=\ham_0 + \ham'$ with $\lambda = 0.3$. In the inset, it is clear that the ground and the first excited state $\ham_0$ cross frequently --- this is due to particle number conservation of $\ham_0$. The nonzero $\lambda$ breaks the symmetry and opens up a gap at each of these level crossings between the ground and the first excited states. However, the gap is relatively small, therefore rate of change of $\mu$
must be small. 

Satisfying the requirements of slow changes in the Hamiltonian, small $\delta t$, and a symmetry breaking field is enabled by the compression algorithm, which allows for arbitrarily long evolution at a fixed depth. Thus, both $\delta t$ and $d\mu/dt$ can be as small as desired
in the Trotter expansion. As a simple first order Trotter circuit without compression, using $\delta t = 0.4$ until $t=800$ and $\delta t = 0.2$ for the remainder, this circuit consists of 3500 Trotter steps and 315,000 CNOTs.  After compression, the circuit is structured as
shown in \cref{fig:cTFIM_blocks} with only 108 CNOTs.

\section{New type of block: Q-Block}\label{sec:newblock_qblock}

In this section, we introduce a new type of block, \textit{Q-block}, which interacts with the blocks via \textit{Q-block rules}. In the end, the usage of $Q$-block together with $B$-blocks will allow us to compress a singly controlled free fermionic evolution with the following Hamiltonian
\begin{align}\label{eq:complete_ham}
\begin{split}
    \ham(t) &= \sum_{i,j} \big( h_{ij}(t) \, c_i^\dagger c_j + p_{ij}(t) \, c_i c_j + q_i(t) c_i\big) \\
    + \mathbf{Z_0} \, & \sum_{\mathbf{i,j}} \mathbf{\big( h'_{ij}(t) \,  c_i^\dagger c_j + p'_{ij}(t) \,  c_i c_j + q'_i(t)  c_i\big)} + \hc\, ,
\end{split}
\end{align}
where the $0$-th qubit is the control qubit. We also show that 1-D spin models TFXY anf TFIM with periodic boundary condition are equivalent to singly controlled free fermions, and provide $Q$-block mappings for these models as well. Finally, we apply our method to calculate the Zak phase of the free Creutz-Hubbard model by calculating the overlap of the ground states of the Creutz-Hubbard model with different parameters that are generated via adiabatic state preparation.

\subsection{Q-block and Q-compression}
\label{sec:qblock-qcompression}

We define a \qblock\ as a mathematical operator that satisfies the following three relations.

\begin{definition}[\qblock]
\label{def:q-block}
Given \bblock s $B_i$ with $i \geq 1$, define a \emph{\qblock}, $Q=Q(\myvec{\theta})$, as an operator that satisfies:
\begin{enumerate}
  \item \textbf{$Q$-fusion:} For any set of parameters $\myvec{\alpha}$ and $\myvec{\beta}$, there exist $\myvec{a}$ such that
        \begin{equation}
          Q(\myvec{\alpha}) \, Q(\myvec{\beta}) = Q(\myvec{a}),
        \end{equation}
  \item \textbf{$Q$-commutation:} For any set of parameters $\myvec{\alpha}$ and $\myvec{\beta}$
        \begin{equation}
          Q(\myvec{\alpha}) \, B_i(\myvec{\beta}) = B_i(\myvec{\beta}) \, Q(\myvec{\alpha}), \qquad i \neq 2,
        \end{equation}
  \item \textbf{$Q$-turnover:} For any set of parameters $\myvec{\alpha}$, $\myvec{\beta}$ and $\myvec{\gamma}$  there exist $\myvec{a}$, $\myvec{b}$ and $\myvec{c}$ such that 
        \begin{equation}
          B_2(\myvec{\alpha}) \, Q(\myvec{\beta}) \, B_2(\myvec{\gamma}) = Q(\myvec{a}) \, B_2(\myvec{b}) \, Q(\myvec{c}).
        \end{equation}
\end{enumerate}
If the $Q$- and $B_i$-blocks satisfy the properties listed above, we will say that $\{Q,B_i\}$ is a \qblock\ mapping.
\end{definition}

The three \qblock\ properties are illustrated in~\cref{fig:all_qblock_things}(a-c). The \qblock s are not explicitly labeled with an index, as they can  only appear in one level, which is at the same height as the first row of \bblock s (see~\cref{fig:all_qblock_things}). \redtext{As can be seen, the \qblock\ and the block with index 1, i.e. $B_1$, satisfy the same properties, and the set of \bblock\ and \qblock\ rules are symmetric under the exchange ${Q} \leftrightarrow B_1$. This property emerges in the examples we provide in \cref{subsec:controlled_free_fermions,subsec:controlled_free_fermions_creation}, where the \qblock\ is mapped to a term that corresponds to the controlled version of $B_1$.}

We are now ready to define a \qdiam, which is the minimal structure for a quantum circuit that admits a \qblock\ mapping.
\begin{definition}[\qdiam]
\label{def:q-diamond}
Define a ``\qdiam'' of height $n$ as,
\begin{align*}
  D^Q_n&({\myvec{\alpha}}) := \\ &\prod_{m=1}^{(n+1)/2} \left[ \prod_{i=n \downarrow}^{1}B_i(\myvec{\alpha}_{i,2m-1})
   \prod_{i=n \downarrow}^{2} B_i(\myvec{\alpha}_{i,2m}) Q(\myvec{\alpha}_{1,2m}) \right],
\end{align*}
for odd $n$, and
\begin{align*}
  D^Q_n&({\myvec{\alpha}}) := \\ & \prod_{m=1}^{n/2} \Bigg[ \prod_{i=n \downarrow}^{1} B_i(\myvec{\alpha}_{i,2m-1})
  \prod_{i=n \downarrow}^{2} B_i(\myvec{\alpha}_{i,2m}) Q(\myvec{\alpha}_{1,2m}) \Bigg] \\ & \prod_{i=n \downarrow}^{1} B_i(\myvec{\alpha}_{i,n+1}),
\end{align*}
for even $n$. Here $\downarrow$ in the product means that the multiplication is done in the decreasing order, $B_i$ are \bblock s, $Q$ is a \qblock\ and each term in the product can have different parameters.
\end{definition}

\begin{figure*}[htpb]
    \centering
    \includegraphics[width = 1.9 \columnwidth]{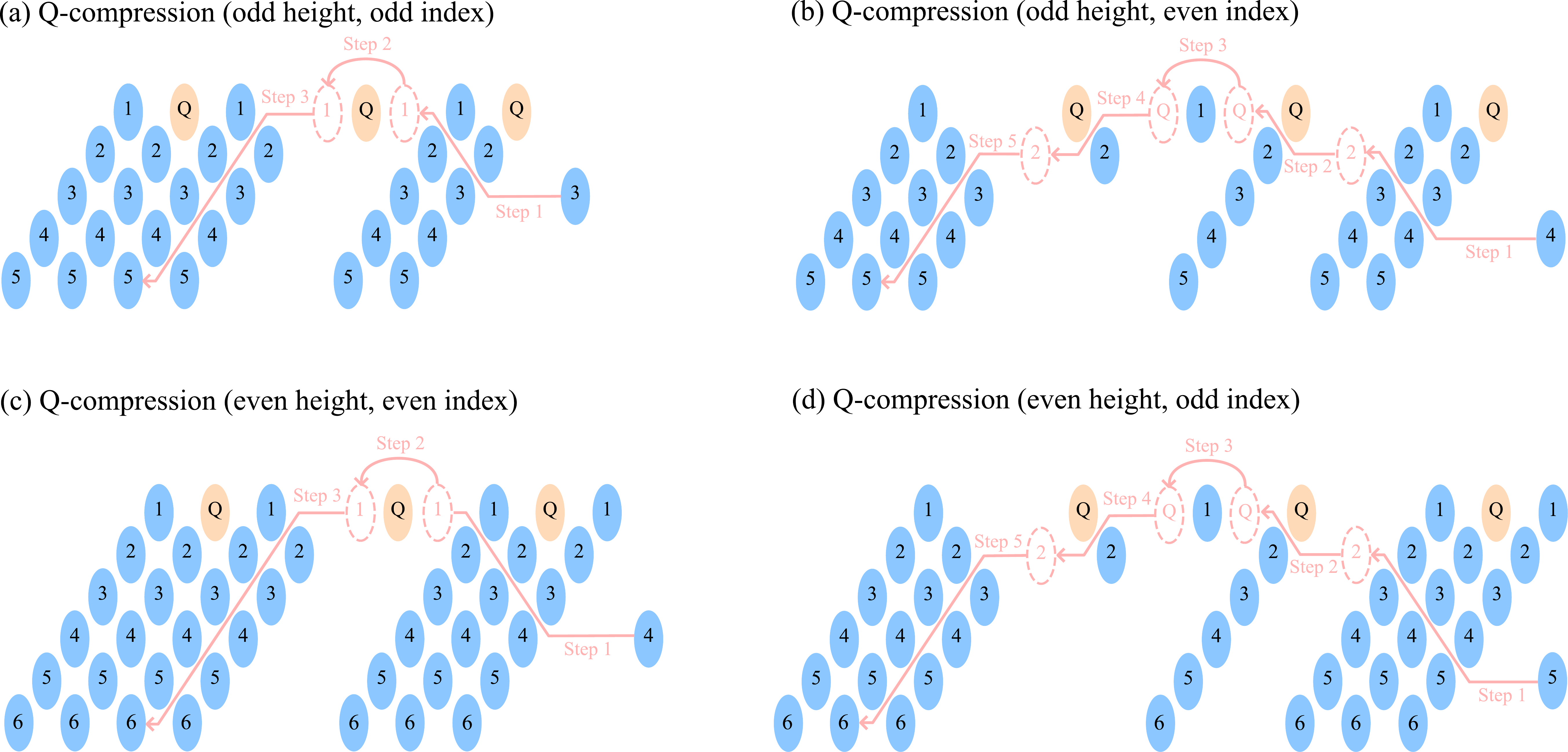}
    \caption{An illustrative proof of \cref{thm:q-compression}, i.e. how the diamond can absorb a block. Because the $Q$-diamond structure differs when the height is even or odd, and the operations to absorb a block changes when the indes $i$ is even or odd, we illustrate each case in different panels. In panels (a-b), the size of the diamond is $n=5$, and the index of the block is $i=3$ on panel (a), and $i=4$ on panel (b). In panels (c-d), the size of the diamond is $n=6$, and the index of the block is $i=4$ on panel (c), and $i=5$ on panel (d). In panels (a) and (c), the block is absorbed by the $Q$-diamond structure only via turnover, $Q$-commutation and fusion operations. On the contrary, the panels (b) and (d) illustrates the cases when a $Q$-turnover is also required.}
    \label{fig:all_qblock_things2}
\end{figure*}

An illustration of a \qdiam\ of height $5$ and $6$ are shown in~\cref{fig:all_qblock_things}(d) and (e). Next, we prove that a \qdiam\ is indeed a minimal realization for a quantum circuit that admits a \qblock\ mapping by showing that the \qdiam\ can absorb any \bblock\ and \qblock.
\begin{theorem}[$Q$-compression]\label{thm:q-compression}
A \qdiam\ of height $n$ can be merged with any \bblock\ $B_i$ with $i=1,2,...,n$ and \qblock\ $Q$,
\begin{align}
\begin{split}
        &D^Q_n({\myvec{\alpha}},{\myvec{\beta}})\,Q(\myvec{\gamma}) = D^Q_n({\myvec{a}},{\myvec{b}}), \\
        &D^Q_n({\myvec{\alpha}},{\myvec{\beta}}) B_i(\myvec{\theta}) = D^Q_n({\myvec{u}},{\myvec{v}}).
\end{split}
\end{align}
\redtext{For odd values of $n$, merging $Q$ with the diamond requires only one $Q$-fusion, while merging $B_1$ requires $n-1$ turnovers and $1$ fusion. For even values of $n$, merging $Q$ with the diamond }requires $n-1$ turnovers and $1$ fusion, while merging $B_1$ requires $1$ fusion. Merging $B_i$ with $i>1$ requires $n-i-2$ turnover and 1 merge operations when $n-i$ is even, and $n-i-4$ turnover, 2 $Q$-turnover and 1 merge operations when $n-i$ is odd.
\end{theorem}

\begin{proof}
The proof is illustrated diagrammatically in \cref{fig:all_qblock_things2}. As can be seen, the trajectory that the newly added block follows is similar for both even-odd index and even-odd height cases, with the only difference being the operations that happen on the first row of the diamond.

When the parity of the height and the index $i$ is the same, as demonstrated in \cref{fig:all_qblock_things2}(a) and (c), the $Q$-compression occurs in 3 steps. \textbf{Step 1:} a \bblock\ with odd index $i$ is moved all the way to the first row and its index becomes 1 by means of $i-2$ \bblock\ turnover operations. \textbf{Step 2:} using the $Q$-commutation relation $B_1 Q = Q B_1$, it can be passed through $Q$-block. \textbf{Step 3:} then it gets pushed all the way down via $n-1$ turnover operations, and is fused with the corresponding bottom-most block. This adds up to $n-i-2$ turnover operations and 1 fusion operation.

When the height and the index have different parities, as demonstrated in \cref{fig:all_qblock_things2}(b) and (d), the $Q$-compression occurs in 5 steps. \textbf{Step 1:} The \bblock\ with index $i$ is first moved to the second row via $i-2$ repeated \bblock\ turnover operations. 
\textbf{Step 2:} Then, it becomes a \qblock\ via a $Q$-turnover operation. \textbf{Step 3:} it passes through the block with index 1 using $Q$-commutation rule. \textbf{Step 4:} it again becomes a \bblock\ with index 2 using a $Q$-turnover operation. \textbf{Step 5:} it moves all the way down using $n-2$ \bblock\ turnover operations, and is fused with the corresponding bottom-most block. This requires $n-i-4$ \bblock\ turnover operations, 2 $Q$-turnover operations, and 1 fusion operation. 
\end{proof}

An alternative type of block, which we denote ``\pblock", its block rules and its compression algorithm is presented in \cref{subsec:pblock_rules_compression}. These rules yield results comparable to those obtained with the \qblock\ rules. However, they are not as generalizable as $Q$-block rules.

\subsection{\qblock\ mapping for Controlled Free Fermions}
\label{subsec:controlled_free_fermions}

In this section, we will provide a $Q$-block mapping that allows us to compress controlled time evolution circuits for any free fermionic model.  
We will first simplify the controlled time evolution of~\cref{eq:complete_ham}  without the annihilation and creation terms, i.e. the evolution under the following Hamiltonian,
\begin{equation}
\begin{split}
    \ham(t) = & \sum_{i,j} \big( h_{ij}(t) \, c_i^\dagger c_j + p_{ij}(t) \, c_i c_j\big) \\
    + \mathbf{Z_0} & \sum_{\mathbf{i,j}} \big( \mathbf{h'_{ij}(t) \, c_i^\dagger c_j + p'_{ij}(t) \, c_i c_j} \big)  + \hc ,
\end{split}
\label{eq:simple_controlled}
\end{equation}
%
where the new terms are highlighted. 
We provide the following  $Q$-block mapping for this Hamiltonian:
\begin{align}\label{eq:simple_controlled_qblocks}
\begin{split}
    B^\TFIM_{2i-1}(\theta) &= e^{i \theta Z_{i}},\\
    B^\TFIM_{2i}(\theta) &= e^{i \theta X_{i} X_{i+1}},\\
    Q^\TFIM(\theta) &= e^{i \theta Z_0 Z_1},
\end{split}
\end{align}
which is illustrated in Fig.~\ref{fig:controlled_TFIM_blocks}(a).
We will show that this mapping indeed covers every interaction term of the Hamiltonian given in \cref{eq:simple_controlled}, and it satisfies all \bblock\ and \qblock\ rules.

Let us show that the set of gates given in \cref{eq:simple_controlled_qblocks} covers all of the new terms given in \cref{eq:simple_controlled}.
When a 1st order Trotter formula is considered, the new terms in bold generate the following structure
\begin{align}\label{eq:chandelier}
\vcenter{\hbox{\includegraphics[width = 0.5\columnwidth]{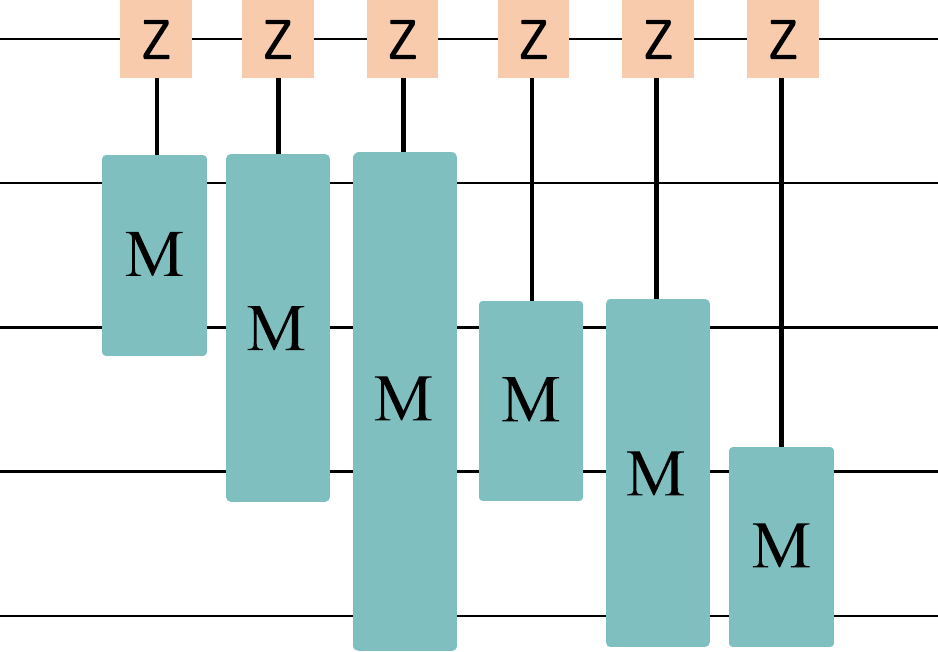}}}\, ,
\end{align}
with the $Z$-rotation on the top qubit, which is the ancilla qubit with index $0$, coupled to the free fermionic gates from \cref{eq:generalized_matchgate} on every pair of system qubits $i,j$
representing the evolution under the new quadratic terms in \cref{eq:simple_controlled}. For example,
\begin{equation}
\begin{aligned}
\label{eq:BTFXY1}
\begin{gathered}
        \includegraphics[width = 0.15\columnwidth, trim={0 0 25cm 0}, clip]{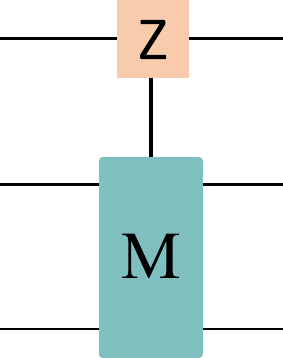}
\end{gathered}
\equiv \ & e^{i \theta_1\: Z_0 Z_1}e^{i \theta_2\: Z_0Z_{2}}e^{i \theta_3\: Z_0X_1 X_{2}}\\ 
         & e^{i \theta_4\: Z_0Y_1 Y_{2}}\:e^{i \theta_5\: Z_0Z_1}e^{i \theta_6\: Z_0Z_{2}},
\end{aligned}
\end{equation}
is a $Z_0$-controlled free fermionic gate on qubits $1$ and $2$. Note that the top qubit, i.e. the ancilla qubit, is labeled as the 0th qubit so that the system qubits start at 1, which is consistent with the $B$-block convention given in \cref{fig:all_block_things,fig:cTFIM_blocks}.

The gate given in \cref{eq:BTFXY1} can be generated via the mapping \cref{eq:simple_controlled_qblocks}. The terms that are not directly in the mapping are rotations with the Pauli matrices $Z_0 Z_2,\:Z_0 X_1 X_2$, and $Z_0 Y_1 Y_2$. Using the following relation
\begin{align*}
Z_0 Z_2 = \fswap_{1,2} \: Z_0 Z_1 \: \fswap_{1,2}^\dagger,
\end{align*}
we can generate $Z_0 Z_2$ rotation via FSWAP gates which are included in $B_i^\TFIM$, and a $Z_0 Z_1$ rotation which is $Q^\TFIM$. Similarly, the following relations allow us to generate the rotations with $Z_0 X_1 X_2$ and $Z_0 Y_1 Y_2$ via the mapping \cref{eq:simple_controlled_qblocks}:
\begin{align*}
Z_0 X_1 X_2 = e^{i \frac{\pi}{4} Z_1} e^{i \frac{\pi}{4} X_1 X_2 }\: Z_0 Z_1 \: e^{-i \frac{\pi}{4} X_1 X_2 } e^{-i \frac{\pi}{4} Z_1 },
\end{align*}
and
\begin{align*}
Z_0 Y_1 Y_2 = e^{i \frac{\pi}{4} Z_2} e^{i \frac{\pi}{4} X_1 X_2 }\: Z_0 Z_1 \: e^{-i \frac{\pi}{4} X_1 X_2 } e^{-i \frac{\pi}{4} Z_2 }.
\end{align*}
Thus, the controlled fermionic gates on qubits 1 and 2 given in \cref{eq:BTFXY1} can be generated via the gates given in \cref{eq:simple_controlled_qblocks}.  

\begin{figure*}[htpb]
    \centering
    \includegraphics[width = 2.1 \columnwidth]{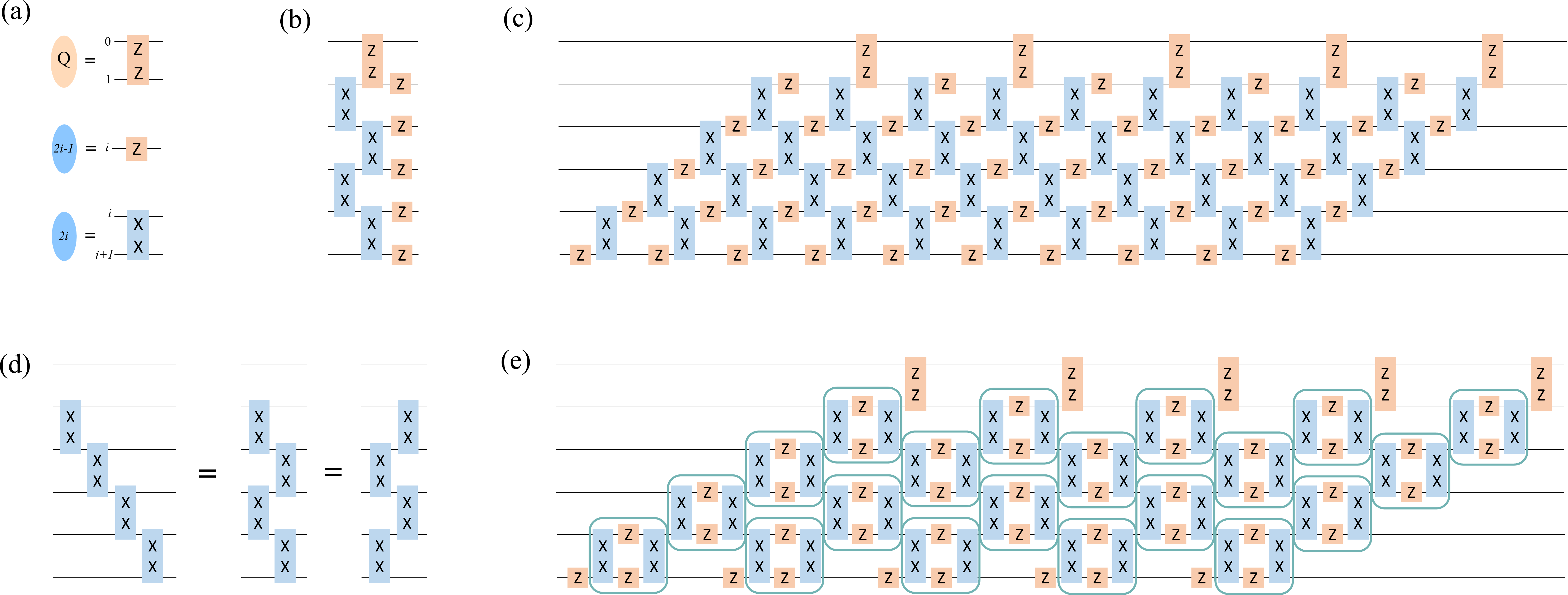}
    \caption{In panel (a), we show the $Q$-block mapping that covers the controlled free fermions, where the control qubit is the $0$th qubit. Panel (b) shows the complete set of blocks and $Q$-block for $n = 5$ system qubits and $1$ control qubit. In panel (c), we show the final circuit representing the diamond structure (see Def. \ref{def:q-diamond}) with this particular $Q$-block mapping, that can be obtained by using the compression theorems. In this form, the circuit requires $4n(n-1)+2n = O(4n^2)$ CNOT gates. By using the relation given in panel (d), certain $XX$ gates can be grouped as shown in panel (e). Then, using the relation \cref{eq:tfim2tfxy}, these groups can be transformed into free fermionic gates. After this simplification, the number CNOTs is reduced to $2n^2 = O(2n^2)$, which is approximately half of the NOT count of the circuit in panel (c).}
    \label{fig:controlled_TFIM_blocks}
\end{figure*}

A generic free fermionic gate on any pair of qubits can be generated via 
the mapping \cref{eq:simple_controlled_qblocks}, by the use of FSWAP gates and \cref{eq:fermion_carry}.  
For example, a $Z_0 c_4 c_6$ rotation can be generated via $Z_0 c_1 c_2$ and FSWAP gates in the following way,
\begin{align}\label{eq:long_range_controlled}
\vcenter{\hbox{\includegraphics[width = 0.75\columnwidth]{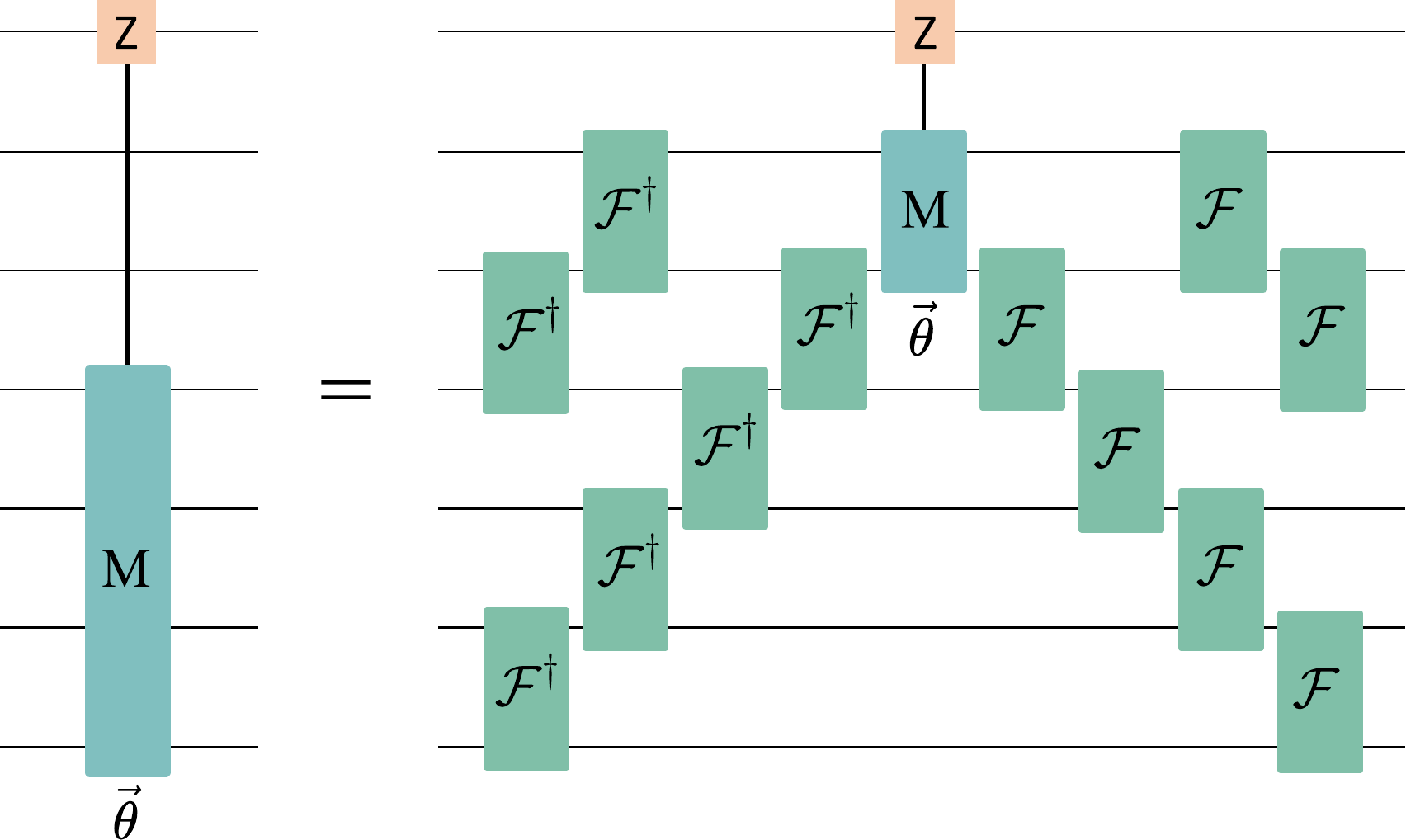}}}\, .
\end{align}
Since both FSWAP gates and the $Z_0 c_1 c_2$ gate can be generated via \cref{eq:simple_controlled_qblocks}, it follows that every term given in the Hamiltonian \cref{eq:simple_controlled} can be generated via the gates given in \cref{eq:simple_controlled_qblocks}.

Let us now show that the mapping $\{Q^\TFIM, B^\TFIM_i\}$ given in \cref{eq:simple_controlled_qblocks} is a $Q$-block mapping.
Because we know that $B^\TFIM_i$ form a block mapping, we only need to show that $Q^\TFIM$ satisfy the $Q$-block mapping rules given in \cref{def:q-block} and \cref{fig:all_qblock_things}(a-c). 
The $Q^\TFIM$ satisfies the $Q$-fusion property because of the following 
\begin{align}
Q^\TFIM(\alpha)Q^\TFIM(\beta) = Q^\TFIM(\alpha + \beta).    
\end{align}
To prove that $Q$-commutation is satisfied, we need to show that $Q^\TFIM$ commutes with all $B_i^\TFIM$ for $i \geq 3$ and $i = 1$. For $i \geq 3$, $Q^\TFIM$ and $B^\TFIM_i$ commute simply because they act on different sets of qubits. For $i = 1$, even though they share a qubit, they still commute since $[Z_0 Z_1, Z_1] = 0$. Thus, the mapping satisfies the $Q$-commutation rule.   
Finally, to confirm the $Q$-turnover relation, we can check the operators in the exponents of $Q^\TFIM$ and $B^\TFIM_2$, i.e., $Z_0 Z_1$ and $X_1 X_2$. These are two anti-commuting Pauli strings, and form the following representation of $\mathfrak{su}(2)$:
\begin{align}\label{eq:su2_relation}
    \mathfrak{su}(2) \equiv i\mathrm{span} \{ Z_0 Z_1, X_1 X_2, Z_0 Y_1 X_2 \}.
\end{align}
The Euler decomposition of this $\mathfrak{su}(2)$ yields
that there exist $a,b,c \in \mathbb{R}$ for any $\alpha, \beta, \gamma \in \mathbb{R}$ such that 
\begin{align}
    e^{i a X_1 X_2} e^{i b Z_0 Z_1} e^{i c X_1 X_2} = e^{i \alpha Z_0 Z_1} e^{i \beta X_1 X_2} e^{i \gamma Z_0 Z_1}, 
\end{align}
which is equivalent to
\begin{align}
\begin{split}
    B^\TFIM_2(a) &Q^\TFIM(b) B^\TFIM_2(c) \\ = &Q^\TFIM(\alpha) B^\TFIM_2(\beta) Q^\TFIM(\gamma), 
\end{split}
\end{align}
where the corresponding angles can be calculated via Eqs. (\bluetext{29}) and (\bluetext{30}) of \cite{kokcu2022algebraic}.
Thus, $Q^\TFIM$ and $B^\TFIM_2$ satisfy the $Q$-turnover property. We conclude that the mapping in \cref{eq:simple_controlled_qblocks} is indeed a \qblock\ mapping.

This \qblock\ mapping enables compression of the controlled free fermionic Hamiltonian given in \cref{eq:simple_controlled}.
Using the $Q$-compression algorithm given in \cref{thm:q-compression}, the fixed depth circuits generated by this mapping for $n = 5$ system qubits and 1 ancilla qubit can be found as \cref{fig:controlled_TFIM_blocks}(c) and (e). The circuit in \cref{fig:controlled_TFIM_blocks}(c) is directly obtained from the diamond structure given in \cref{def:q-diamond}. It contains $n$ $ZZ$-rotations and $2n(n-1)$ $XX$-rotations, which leads to $4n^2-2n$ CNOT gates. The number of CNOT gates can be reduced by moving $X_i X_{i+1}$ gates around as shown in \cref{fig:controlled_TFIM_blocks}(d). This reordering combines certain $XX$ gates, and transforms the circuit from \cref{fig:controlled_TFIM_blocks}(c) to (e), where certain $XX$ gates are grouped together. We can then use \cref{eq:tfim2tfxy}, and transform these groups to TFXY blocks or free fermionic gates, which can be implemented by only 2 CNOTs instead of 4. Because there are $n$ $ZZ$ gates and $n(n-1)$ groups, the final CNOT count for the circuit in panel (e) becomes $2n + 2n(n-1) = 2n^2$, which is approximately half the CNOT count of the circuit in panel (c). 

The same compression can also be achieved with the alternative $P$-block mapping given in \cref{subsec:pblock_controlled_free_fermions}.
This method is capable of generating a CNOT efficient circuit without the usage of TFIM $\leftrightarrow$ TFXY transformation, and is more efficient in the case of compression of a few elements due to the $O(n^2)$ overhead of the transformation. However, $Q$-compression is faster because it is based on TFIM compression, and $P$-compression if based in TFXY compression \cite{kokcu2022fixed,camps2022algebraic}.

{
We apply our \qblock\ mapping and $Q$-compression algorithm to the calculation of a topological phase in the imbalanced Creutz-Hubbard model, which is shown in~\cref{fig:Creutz-Hubbard}. The model has two groups of fermions labeled by $\ell \in \{u,d\}$, with on-site energy terms and hopping both within and between the $u,d$ groups.
The Hamiltonian
of the model is given as the following  \cite{junemann2017exploring}
\begin{align}\label{eq:Creutz_ham}
\begin{split}
    \ham =& \sum_{\substack{i,\ell}} \tilde{t}\Big( - c_{i+1, \ell}^\dagger c_{i, \Bar{\ell}} + i s_\ell c_{i+1, \ell}^\dagger c_{i, \ell} + \hc \Big)\\
    & + \sum_{\substack{i,\ell}} \frac{\Delta}{2} s_\ell \: c_{i, \ell}^\dagger c_{i, {\ell}},
\end{split}
\end{align}
where $\ell \in \{u,d\}$ and $\Bar{\ell}$ is the opposite choice, i.e. if $\ell = u$, then $\Bar{\ell} = d$ and vice versa. We set $s_u = 1$ and $s_d = -1$. The term proportional to $\Delta$ creates an imbalance between two parts of the ladder. 
}

\begin{figure}[t]
\centering
\begin{tikzpicture}[
	node1/.style={circle, draw, thick, minimum size=0.4cm, color=myred},
	node2/.style={circle, draw, thick, minimum size=0.4cm, color=mygreen, fill=black!10}
]
\pgfmathsetmacro{\yos}{-1.5}
\pgfmathsetmacro{\dx}{1.2}

  \foreach \x in {1,2,...,3} {
  	\pgfmathtruncatemacro{\label}{Mod(2*(\x-1),8)};
    \node[node1] (node1\x) at (\dx*\x,0) {$\label$};
  }
  \node[node1, dotted] (node15) at (\dx*4,0) {$0$};
  
  \foreach \x in {1,2,...,3} {
    \pgfmathtruncatemacro{\label}{Mod(2*\x-1,8)};
    \node[node2] (node2\x) at (\dx*\x,\yos) {$\label$};
  }
  \node[node2, dotted] (node25) at (\dx*4,\yos) {$1$};
  
  \draw[color=myred, thick] (node11) -- (node12) node[pos=0.5,above] {\scriptsize{$\mathrm{i} s_u \tilde{t}$}};
  \draw[color=mygreen, thick] (node21) -- (node22) node[pos=0.5,below] {\scriptsize{$\mathrm{i} s_d \tilde{t}$}};
  \foreach \x in {2} {
    \draw[color=myred, thick] (node1\x) -- (node1\number\numexpr\x+1\relax);
    \draw[color=mygreen, thick] (node2\x) -- (node2\number\numexpr\x+1\relax);
  }
  \draw[dotted,color=myred, thick] (node13) -- (node15);
  \draw[dotted,color=mygreen, thick] (node23) -- (node25);
  
  \draw[color=myred, thick] (node11) to [in=120, out=60, loop] node[above] {\scriptsize{$s_u \Delta/2$}} (node);
  \draw[color=mygreen, thick] (node21) to [in=300, out=240, loop] node[below] {\scriptsize{$s_d \Delta/2$}} (node);
  \foreach \x in {2,...,3} {
	\draw[color=myred, thick] (node1\x) to [in=120, out=60, loop] node[above] {} (node);
    \draw[color=mygreen, thick] (node2\x) to [in=300, out=240, loop] node[below] {} (node);
  }
  \draw[dotted, color=myred, thick] (node15) to [in=120, out=60, loop] node[above] {} (node);
  \draw[dotted,  color=mygreen, thick] (node25) to [in=300, out=240, loop] node[below] {} (node);
  
  \foreach \x in {1,2}{
  	\draw[thick] (node1\x) -- (node2\number\numexpr\x+1\relax);
	\draw[thick] (node2\x) -- (node1\number\numexpr\x+1\relax);
  }
  \draw[dotted, thick] (node13) -- (node25);
  \draw[dotted, thick] (node23) -- (node15);
  
  \node[left of=node11] {$\bm{\ell = u}$};
  \node[left of=node21] {$\bm{\ell = d}$};
  \node[below= 0.15cm of node11] {\scriptsize{$-\tilde{t}$}};
  
\end{tikzpicture}    
\caption{The 6-site free Creutz-Hubbard model. The terms in the Hamiltonian connecting
different sites, i.e. allowing fermions to hop between different sites, are indicated by lines (dashed lines indicate the periodic terms). Self connecting lines correspond to chemical potential terms, parallel lines correspond to the tunnelling terms with phase $i s_\ell$, and the diagonal black lines correspond to the terms with no phase.}
\label{fig:Creutz-Hubbard}
\end{figure}

\begin{figure*}[htpb]
    \centering
    \includegraphics[width = 2.1\columnwidth]{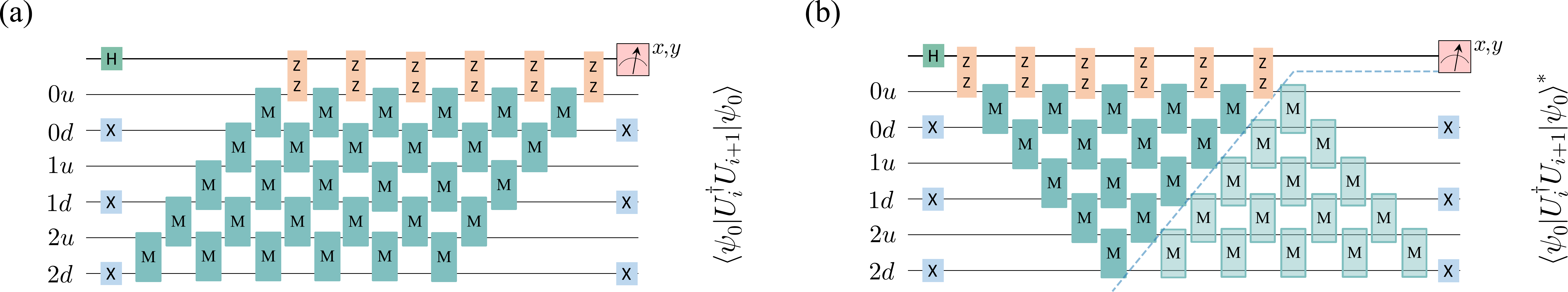}
    \caption{
    (a) The Hadamard test circuit to calculate the free fermionic overlap. The controlled free fermionic evolution is compressed down to a \qdiam structure as given in \cref{fig:controlled_TFIM_blocks}(e), where we did replace the groups with the free fermionic gates. (b) By just taking the complex conjugation of the overlap, the circuit can be inverted. Because the Hadamard test requires only the measurement of the ancilla qubit, approximately half of the circuit (the free fermionic gates below the diagonal dashed line) can be neglected.
    }
    \label{fig:simplified_diamond}
\end{figure*}

{
The model has a topological phase transition at the value of $\Delta = 4 \tilde{t}$, and the order parameter of this transition is called the ``Zak phase" \cite{junemann2017exploring}. The Zak phase is defined as the following
\begin{align}
    \phi_{\mathrm{Zak}} := 
    \int_{\mathrm{BZ}} dq\: \braket{q|\partial_q|q}, 
\end{align}
where $\ket{q}$ is the single particle state with quasi-momentum value $q$, and the integral is taken over the Brillouin zone $[-\pi,\pi]$. For the Hamiltonian given in \cref{eq:Creutz_ham}, the Zak phase can be analytically calculated to be
\begin{align}
    \phi_{\mathrm{Zak}} = \pi \theta(4 \tilde{t} - \Delta).
\end{align}
For $\Delta < 4 \tilde{t}$, the Zak phase is non-zero, and the model shows a topologically non-trivial behavior \cite{junemann2017exploring}.
}

{
Numerical calculations of the Zak phase require a discretized integral
over the Brillouin zone, which is problematic when only a few
sites are used; the available momenta for an $N$-site periodic lattice
are
$k_n = 2\pi n/N$ with $n = 0,1,\dots,N-1$. 
This problem can be circumvented by employing twisted boundary conditions~\cite{niu1985quantized,xiao2023robust}. By modifying the boundary hopping
terms with an angle $\exp{i\varphi}$
the allowed momentum values become $k_n(\varphi) = 2\pi n/N + \varphi/N$. Thus, if we vary the twist angle from $0$ to $2\pi$, we run through all momentum values, without requiring a larger unit cell, and while remaining in real space. Measuring the Zak phase then can be done by multiplying the overlap amongst the ground states of the Creutz-Hubbard Hamiltonian with different twist angels.  
}
{
To obtain the ground states, we use adiabatic evolution which starts from
the ground state of a Creutz-Hubbard model with $\Delta \gg \tilde{t}$. In this case, the imbalance term dominates, and the ground state becomes very close to the state $\ket{\psi_0} = \ket{0101\dots0101}$, where the $d$ sites are occupied and $u$ sites are empty. The $\Delta$ value is then adiabatically evolved to the desired value, to obtain
the ground state of $\ham(\varphi = 0)$ with twist angle $\varphi = 0$. Then the twist angle is adiabatically evolved to $\varphi = 2 \pi$.
\rredtext{By doing so, we produce the ground state for a given set of twist angle values $\{ \varphi_i \}$, and the desired imbalance $\Delta$ value. The Zak phase can then be approximated by the following:}
\begin{align}\label{eq:numerical_Zak}
    e^{i \phi_{\mathrm{Zak}}} = \prod_i \braket{\varphi_i |\varphi_{i+1}},
\end{align}
where $\ket{\varphi_i}$ is the ground state of the Creutz-Hubbard Hamiltonian $\ham(\varphi_i)$ with twist angle $\varphi_i$. \rredtext{This approximation becomes exact when $\varphi_i$ form a uniform continuous set.} 
}

\rredtext{
Let us define a unitary $U_i$ such that $\ket{\varphi_i} = U_i \ket{\psi_0}$ where $\ket{\psi_0} = \ket{0101\dots0101}$ is the ground state of the imbalance term. Since we generate the state via adiabatic time evolution, this $U_i$ consists of free fermionic evolution. By using the TFXY block mapping given in \cref{fig:all_block_things}(g) and the compression given in \cref{fig:all_block_things}(d), this free fermionic evolution can be compressed down to a TFXY triangle. The transition amplitude given in \cref{eq:numerical_Zak} reads $\braket{\varphi_i |\varphi_{i+1}} = \braket{\psi_0|U_i^\dagger U_{i+1}|\psi_0}$. Since both $U_i$ and $U_{i+1}$ can be represented as a TFXY triangle, $U_i^\dagger U_{i+1}$ also can be compressed into a TFXY triangle. This transition amplitude can be calculated via the Hadamard test. In order to do so, one needs to implement a circuit in which the back and forth evolution $U_i^\dagger U_{i+1}$ is implemented if and only if the ancilla qubit is in state $\ket{1}$, which requires controlled free fermionic evolution, and can be compressed into a diamond circuit as shown in \cref{fig:controlled_TFIM_blocks}.
}

\begin{figure}[b]
    \centering
    \includegraphics[width = 1.0\columnwidth]{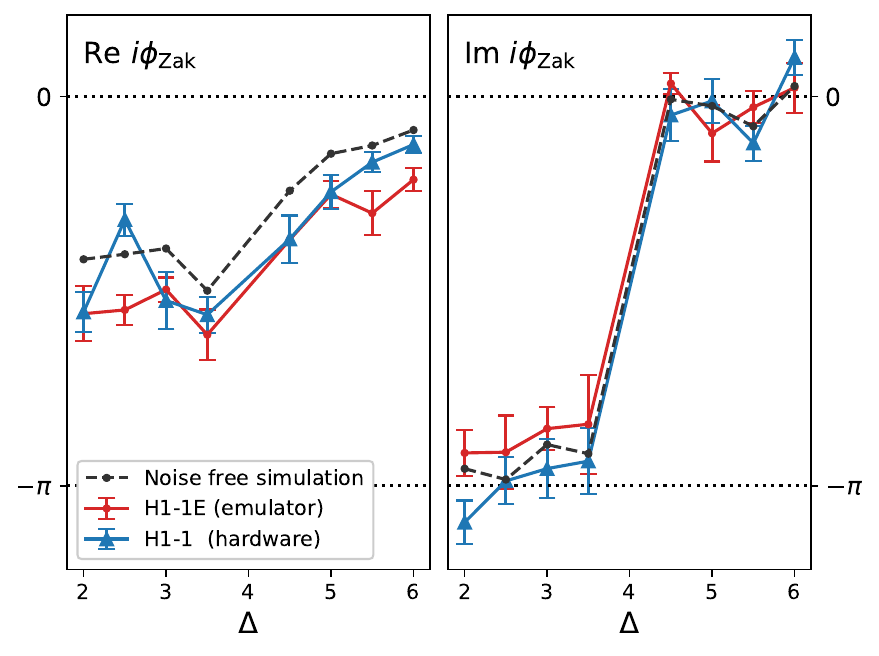}
    \caption{
    Real and imaginary parts of $i\phi_\mathrm{Zak}$
    obtained via \cref{eq:numerical_Zak}, using 5 total
    states.
    The imaginary part shows a topological phase transition,
    switching from $-\pi$ to $0$ at $\Delta=4$. The real
    part is non-zero due to the finite twist angle step $d\phi = 2\pi/5$ (see text for more details).
    }
    \label{fig:Zak_phase_plot}
\end{figure}

\rredtext{After obtaining the \qdiam\ for  the overlap $\braket{\psi_0|U^\dagger_i U_{i+1}|\psi_0}$, shown in \cref{fig:simplified_diamond}(a),
we can invert the circuit to instead compute the overlap $\braket{\psi_0|U^\dagger_i U_{i+1}|\psi_0}^*$ (\cref{fig:simplified_diamond}(b)) from which we can obtain the Zak phase via complex conjugation. As we are only measuring the ancilla qubit in a Hadamard test, we can discard the gates that have no effect on the measurement result of the inverted circuit, thereby reducing the number of CNOT gates by half.
}
\rredtext{Note that for computing the overlap in any free fermionic calculation, the Hadamard test circuit will have the same \qdiam\ form, and can be simplified via this complex conjugation method.}

\rredtext{
Our detailed protocol is as follows.
We first initialize our system with a large value of $\Delta_0 = 200 \tilde{t}$. In this case, the imbalance term dominates, and the ground state is the half-filled state where all the particles live in the down lattice sites labelled $d$. We adiabatically change $\Delta_0$ to the target value $\Delta$ for our Zak phase calculation, giving us the ground state for $\varphi = 0$. Then, we change the $\varphi$ angle from $0$ to $2\pi$. Classical simulations show that it is sufficient to use 5 values for $\varphi$,  i.e. using $d\varphi = 2\pi/5$. For all the evolution, the controlled compression is used,
and thus we have a diamond structure (c.f. Fig.~\ref{fig:controlled_TFIM_blocks}). We use these diamond structures to calculate the overlap between adjacent ground states as in the right hand side of \cref{eq:numerical_Zak}, which when multiplied yields the Zak phase for the $\Delta$.
}

\rredtext{
We first used Quantinuum's H1-1E emulator which closely mimics the H1 QPU hardware, to simulate the Zak phase of the 6-site Creutz-Hubbard model using 7 qubits and 1,000 shots to obtain the necessary statistics. The results were found to be in good agreement with the analytical results. We subsequently used Quantinuum's H1-1 quantum hardware to compute the Zak phase for the different $\Delta$ values~\footnote{Quantinuum H1-1. \href{https://www.quantinuum.com/}{https://www.quantinuum.com/}, October 26-31, 2023.}. Here we used only 200 shots to extract the Zak phase. Both the simulator and hardware results
are shown together with noise free simulation results in \cref{fig:Zak_phase_plot}. We observe that the imaginary part of the phase exhibits a jump from $-\pi$ to $0$ across $\Delta=4\tilde t$, illustrating that the topological phase transition is properly captured in these calculations. The results from the hardware are in close agreement with the emulator results and noise free results, providing confidence that the Zak phase can be accurately obtained with trapped ion quantum computers. The real part of the phase encodes the decay in the wavefunction overlap in \cref{eq:numerical_Zak}. In the limit $d\varphi \rightarrow 0$ and with perfect hardware, the real part should be identically zero. Here, due to finite $d\varphi$ and hardware noise,
some decay is present, but this does not affect the calculation
of the geometric Zak phase, i.e. $\mathrm{Im} \: i\phi_{\mathrm{Zak}}$.
}

\begin{figure*}[htpb]
    \centering
    \includegraphics[width = 2.1 \columnwidth]{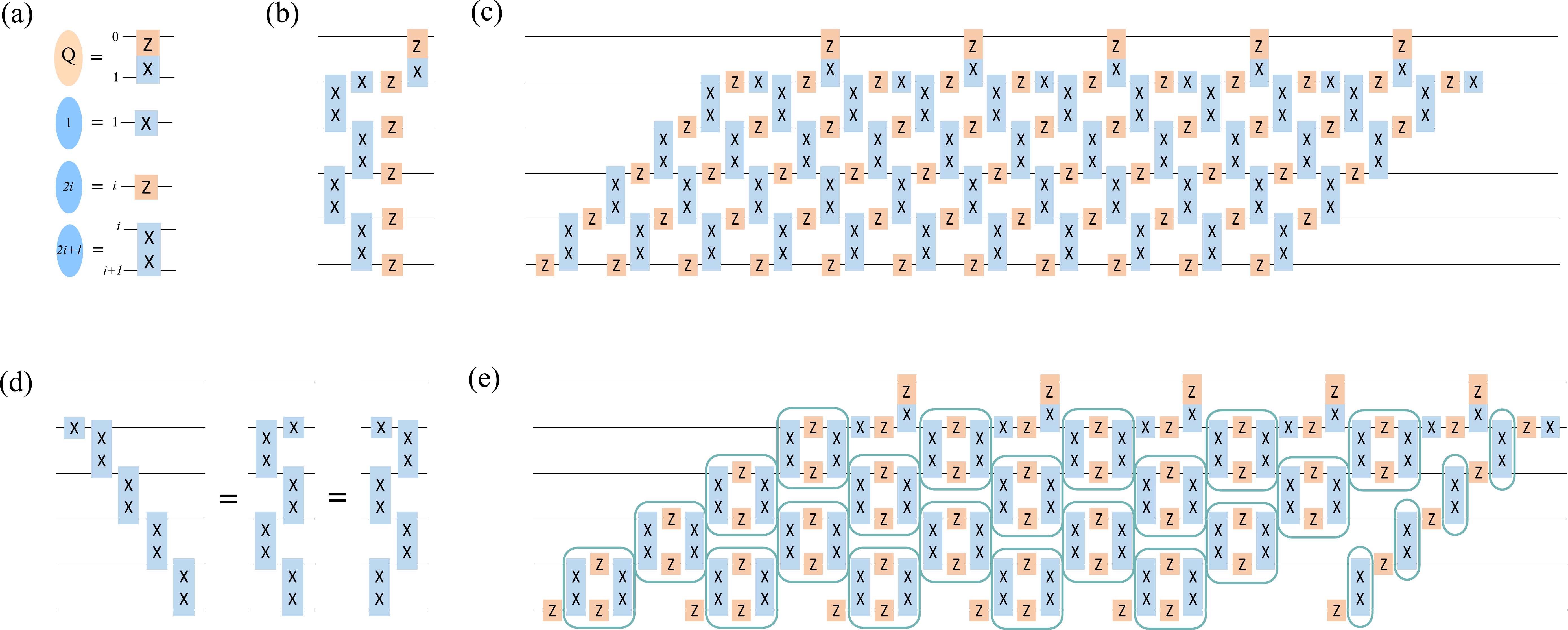}
    \caption{In panel (a), we show the $Q$-block mapping that covers the controlled free fermions with creation, where the control qubit is the $0$th qubit. Panel (b) shows the complete set of blocks and $Q$-block for $n = 5$ system qubits and $1$ control qubit. In panel (c), we show the final circuit representing the diamond structure (see Def. \ref{def:q-diamond}) with this particular $Q$-block mapping, that can be obtained by using the compression theorems. In this form, the circuit requires $4n(n-1)+4n-2 = O(4n^2)$ CNOT gates. By using the relation given in panel (d), certain $XX$ and $X$ gates can be grouped as shown in panel (e). Then, using the relation \cref{eq:tfim2tfxy}, these groups can be transformed into free fermionic gates. After this simplification, the number CNOTs is reduced to $2n^2+2n-2 = O(2n^2)$, which is approximately half of the NOT count of the circuit in panel (c).}
    \label{fig:controlled_TFIM_with_creation_blocks}
\end{figure*}

\subsection{\qblock\ mapping for Controlled Free Fermions with Creation}
\label{subsec:controlled_free_fermions_creation}

In this section, we will provide a $Q$-block mapping that allows us to compress controlled time evolution circuits for any free fermionic model, including the ones with creation and annihilation operators as discussed previously in \cref{sec:creation_annihilation} for the uncontrolled case. In this setting, the system can be described by the following Hamiltonian
\begin{equation}
\begin{split}
    \ham(t) &=  \sum_{i,j} \big( h_{ij}(t) \, c_i^\dagger c_j + p_{ij}(t) \, c_i c_j\big) + \sum_\mathbf{i} \mathbf{q_i(t) c_i} \\
    &+ {Z_0} \sum_{{i,j}} \big( {h'_{ij}(t) \, c_i^\dagger c_j + p'_{ij}(t) \, c_i c_j} \big)  \\
    &+ \mathbf{Z_0}  \sum_{\mathbf{i}}  \mathbf{q'_i(t) c_i} + \hc,
\end{split}
\label{eq:creation_controlled}
\end{equation}
where the last $\hc$ implies the hermitian conjugate of the entire right-hand side. The new terms compared to the Hamiltonian given in \cref{eq:simple_controlled} are creation/annihilation and controlled creation/annihilation terms, and are denoted with bold font. We provide the following \qblock\ mapping to compress the evolution under this Hamiltonian into a \qdiam:
\begin{align}\label{eq:ctfim_qmapping}
\begin{split}
    B^\CTFIM_1(\theta) &= e^{i \theta X_{1}},\\
    B^\CTFIM_{2i}(\theta) &= e^{i \theta Z_{i}},\\
    B^\CTFIM_{2i+1}(\theta) &= e^{i \theta X_{i} X_{i+1}},\\
    Q^\CTFIM(\theta) &= e^{i \theta Z_0 X_1},
\end{split}
\end{align}
where the superscript CTFIM stands for creation TFIM, implying that this mapping contains the fermionic creation and annihilation terms, and the free fermion terms are implemented via TFIM gates. An illustration of this mapping is given in \cref{fig:controlled_TFIM_with_creation_blocks}(a).

Let us show that every term in the Hamiltonian given in \cref{eq:creation_controlled} can be implemented via the gates given in \cref{eq:ctfim_qmapping}.
As we know from \cref{sec:creation_annihilation}, the blocks $B^\CTFIM_i$ cover all the uncontrolled terms, i.e. the terms that do not have a $Z_0$ attached. The following relation of the exponent of $Q^\CTFIM$ show that the mapping \cref{eq:ctfim_qmapping} can implement creation/annihilation of a fermion on site 1 in a controlled fashion:
\begin{align}
\begin{split}
B^\CTFIM_2(\theta) \: Z_0 X_1 \:& B^\CTFIM_2(-\theta) \\ 
=& \cos 2\theta \: Z_0 X_1 + \sin 2\theta \: Z_0 Y_1  \\ 
=& Z_0 \left(e^{2i\theta} c_1 + e^{-2i\theta} c_1^\dagger \right).
\end{split}
\end{align}
With the usage of FSWAP gates, we can then cover all controlled creation/annihilation terms with the mapping in \cref{eq:ctfim_qmapping}. The following relation shows that $Q^\TFIM$ from the TFIM \qblock\ mapping in \cref{eq:simple_controlled_qblocks} can be generated via our new CTFIM mapping: 
\begin{align}
    Z_0 Z_1 = e^{i\frac{\pi}{4} X_1} e^{i\frac{\pi}{4}Z_1} Z_0 X_1 e^{-i\frac{\pi}{4} Z_1} e^{-i\frac{\pi}{4} X_1},
\end{align}
which yields
\begin{align}
\begin{split}
    Q^\TFIM(\theta) =& B^\CTFIM_1(\pi/4) B^\CTFIM_2(\pi/4) Q^\CTFIM(\theta)\\
    &B^\CTFIM_2(-\pi/4)B^\CTFIM_1(-\pi/4).
\end{split}
\end{align}
Since $B^\CTFIM_i$ contains all $B^\TFIM_i$ as well,
the mapping \cref{eq:ctfim_qmapping} then contains the \qblock\ mapping given in \cref{eq:simple_controlled_qblocks}, and thus can implement all controlled fermion pair operations $Z_0 c_i^{(\dagger)} c_j^{(\dagger)}$ in the Hamiltonian \cref{eq:creation_controlled} as well. 

Let us show that the CTFIM mapping is a \qblock\ mapping. From \cref{sec:creation_annihilation}, we know that the CTFIM blocks $B^\CTFIM_i$ satisfy the block properties and form a \bblock\ mapping. We only need to show that $Q^\CTFIM$ satisfy the \qblock\ properties given in \cref{fig:all_qblock_things}(a-c). The $Q$-fusion property is satisfied due to the following: 
\begin{align}
Q^\CTFIM(\alpha)Q^\CTFIM(\beta) 
= Q^\CTFIM(\alpha + \beta).    
\end{align}
To prove that $Q$-commutation is satisfied, we need to show that $Q^\CTFIM$ commutes with all $B_i^\CTFIM$ for $i \geq 3$ and $i = 1$. For $i \geq 4$, $Q^\CTFIM$ and $B^\CTFIM_i$ commute simply because they act on different sets of qubits. For $i = 1$ and $3$, even though they share a qubit, they still commute since $[Z_0 X_1, X_1] = 0$ and $[Z_0 X_1, X_1 X_2] = 0$ respectively. Thus, the mapping satisfy the $Q$-commutation rule.   
Finally, the exponents of $Q^\CTFIM$ and $B^\CTFIM_2$, which are $Z_0 X_1$ and $Z_1$, form the following representation of $\mathfrak{su}(2)$:
\begin{align}\label{eq:su2_relation_creation}
    \mathfrak{su}(2) \equiv i\mathrm{span} \{ Z_0 X_1, Z_1, Z_0 Y_1 \}.
\end{align}
The Euler decomposition of this $\mathfrak{su}(2)$ yields
that there exist $a,b,c \in \mathbb{R}$ for any $\alpha, \beta, \gamma \in \mathbb{R}$ such that 
\begin{align}
    e^{i a X_1 X_2} e^{i b Z_0 Z_1} e^{i c X_1 X_2} = e^{i \alpha Z_0 Z_1} e^{i \beta X_1 X_2} e^{i \gamma Z_0 Z_1}, 
\end{align}
which is equivalent to
\begin{align}
\begin{split}
    B^\TFIM_2(a) &Q^\TFIM(b) B^\TFIM_2(c) \\ = &Q^\TFIM(\alpha) B^\TFIM_2(\beta) Q^\TFIM(\gamma), 
\end{split}
\end{align}
where the corresponding angles can be calculated via Eqs. (\bluetext{29}) and (\bluetext{30}) of \cite{kokcu2022algebraic}.
Thus, $Q^\CTFIM$ and $B^\CTFIM_1$ satisfy the $Q$-turnover property as well. We conclude that the mapping given in \cref{eq:ctfim_qmapping} is indeed a \qblock\ mapping, and by \cref{thm:q-compression}, time evolution under the Hamiltonian in \cref{eq:creation_controlled} can be compressed into a \qdiam. 

\cref{fig:controlled_TFIM_with_creation_blocks} illustrates the \qblock\ mapping and the fixed depth circuits obtained by the $Q$-compression theorem \cref{thm:q-compression}. \cref{fig:controlled_TFIM_with_creation_blocks}(a) illustrates \cref{eq:ctfim_qmapping} diagrammatically, where on the left hand side we have the blocks, and on the right hand side we have the gates. \cref{fig:controlled_TFIM_with_creation_blocks}(b) illustrates all blocks and \qblock\ on $n=5$ system and 1 ancilla qubits for the Hamiltonian in \cref{eq:creation_controlled}. As it can be seen, there is 1 \qblock\ and $2n = 10$ blocks present for this model and \qblock\ mapping.   
In panel (c), we show the final circuit representing the diamond structure (see Def. \ref{def:q-diamond}) with this particular $Q$-block mapping, that can be obtained by using the compression theorems. In this form, the circuit requires $4n(n-1)+4n-2 = O(4n^2)$ CNOT gates. By using the relation given in panel (d), certain $XX$ and $X$ gates can be grouped as shown in panel (e). Then, using the relation \cref{eq:tfim2tfxy}, these groups can be transformed into free fermionic gates. After this simplification, the number CNOTs is reduced to $2n^2+2n-2 = O(2n^2)$, which is approximately half of the CNOT count of the circuit in panel (c). 

This mapping illustrates another advantage of $Q$-blocks over $P$-blocks given in \cref{subsec:pblocks}, because $P$-blocks cannot be generalized to compress the time evolution under the Hamiltonian given in \cref{eq:creation_controlled}.

\section{Discussion and Outlook}
\label{sec:discussion}

With the developments in this paper, we have significantly extended the class of Hamiltonians
whose evolution may be compressed. Specifically, any time dependent mean field Hamiltonian on an arbitrary lattice can now be simulated efficiently, such as the tight binding model with any particle number, superconducting free fermions, free fermions with particle injection, controlled free fermions on any lattice, and as given in \cref{asec:pbc}, certain spin models such as TFIM and TFXY on a ring. 
In addition to that, we can add and remove particles on a given mode with any amplitude. In Ref.~\cite{kokcu2023linear} it was shown that for momentum conserving Hamiltonians, obtaining Green's functions directly in the momentum basis via perturbing the system with $c^\dagger_k$ yields less noisy results compared to obtaining perturbing with $c^\dagger_r$ and post-processing via Fourier transformation. 
Accordingly, we provide a direct method for creating one or more particles with definite momenta in \cref{asec:creating_fermions}.
For other systems, such as those that arise in chemistry, these single particle modes are more general. Our methods can easily be generalized to these modes as well, and be used to compress particle creation/annihilation operators on any mode.
With the addition of controlled evolution, these techniques are now applicable in a broader regime. Controlled evolution is a key step in quantum phase estimation, and similarly plays a role in computing Hamiltonian matrix elements for real time subspace expansions~\cite{klymko2022real,shen2023estimating,shen2024efficient}.

The compression algorithm discussed here is not limited to compressing time evolution. Rather, it may be applied to any set
of quantum gates that obey the \bblock\ and \qblock\ properties.  As long as the block mapping can be found, the algorithms developed here can be readily applied.  We do not expect that general state preparation or other similar algorithms can be fully compressed, as the compressed circuits lack expressibility; the states that can be reached are  limited~\cite{kokcu2022fixed,d2007introduction}. However, there may be sizable subsets of the full quantum circuits whose elements do obey the block properties, and may be significantly shortened. We expect that the developments made in this work can thus have significant impact in transpiler software.

\rredtext{
The compression software is available as part of the fast free fermion compiler (\texttt{F3C}) \cite{f3c, f3cpp} at \url{https://github.com/QuantumComputingLab}.
\texttt{F3C} is based on the \texttt{QCLAB} toolbox \cite{qclab,qclabpp} for creating and
representing quantum circuits.
}

\begin{acknowledgments}
EK and AFK were supported by the National Science Foundation under award No.~1818914: PFCQC: STAQ: Software-Tailored Architecture for Quantum co-design
and No.~2325080: PIF: Software-Tailored Architecture for Quantum Co-Design.
LBO, RVB, and WAdJ were supported by the U.S. Department of Energy (DOE) under Contract No.  DE-AC02-05CH11231,  through the Office of Advanced Scientific Computing  Research  Accelerated  Research  for  Quantum Computing  Program.
This research used resources of the National Energy Research Scientific Computing Center (NERSC), a U.S. Department of Energy Office of Science User Facility located at Lawrence Berkeley National Laboratory, operated under Contract No. DE-AC02- 05CH11231.
This research used
resources of the Oak Ridge Leadership Computing Facility, which is a
DOE Office of Science User Facility supported under Contract
No.~DE-AC05-00OR22725.
We acknowledge the use of IBM Quantum services for this work.
Finally, we acknowledge the use of the QISKIT software package for use in the quantum computer calculations~\cite{qiskit_shorter}.
\end{acknowledgments}

\bibliography{refs}
\bibliographystyle{apsrev4-2}

\clearpage
\onecolumngrid
\appendix

\renewcommand\thefigure{S\arabic{figure}}  
\setcounter{figure}{0}

\section{Fermion Creation with Definite Momentum}
\label{asec:creating_fermions}

One of the interesting features we can exploit from the compressibility of the evolution under \cref{eq:lrfermions_c1} is that it allows us to create/annihilate particles in any single particle mode in a unitary fashion. Here, we will specifically show how to generate a circuit that creates a particle with a definite momentum $k$. 

Creation operators in momentum space $c^\dagger_k$ are a discrete Fourier transformation of creation operators in position space $c^\dagger_r$:
\begin{align}\label{eq:momentum_creator_app}
    c^\dagger_k = \sum_{r=1}^n e^{2 \pi ikr/n} c^\dagger_r,
\end{align}
where $k = 0,1,2,...,n-1$ and $n$ is the number of lattice sites. The $c_k$ satisfy the anti-commutation relations $\{c_p,c_k\} = \{c_p^\dagger,c_k^\dagger\} = 0$ and $\{c_p,c_k^\dagger\} = \delta_{pq}$ where $\delta_{pq}$ is the Kronecker delta. 

The state we would like to create is $\ket{\psi_k} = c^\dagger_k \ket{0}$ where $\ket{0}$ represents the empty fermion state, which is $\ket{000...0}$ in the computational basis. Unfortunately $c_k^\dagger$ is not a unitary operator. However, $c^\dagger_k + c^{\phantom{\dagger}}_k$ is Hermitian and due to the anti-commutation relations, we have $\big( c^\dagger_k + c^{\phantom{\dagger}}_k  \big)^2 = 1$. Moreover, $c_k \ket{0} = 0$. Combining these, we have
\begin{align}
    e^{i \theta \big( c^\dagger_k + c^{\phantom{\dagger}}_k  \big)} \ket{0} = \cos(\theta) \ket{0} + i \sin(\theta) \big( c^\dagger_k + c^{\phantom{\dagger}}_k  \big) \ket{0} = \cos(\theta) \ket{0} + i \sin(\theta) \ket{\psi_k}.
\end{align}
For $\theta = \pi/2$, we obtain $\psi_k$ up to a global phase
\begin{align}
    e^{i \frac{\pi}{2} \big( c^\dagger_k + c^{\phantom{\dagger}}_k  \big)} \ket{0} = i \ket{\psi_k}.
\end{align}
Thus, if we implement time evolution unitary under $\ham_k = c^\dagger_k + c^{\phantom{\dagger}}_k $ 
for time $ t = -\pi/2$, we can create a particle with momentum $k$. Now,
this Hamiltonian is a special case for the Hamiltonian in \cref{eq:lrfermions_c1}, and therefore we can generate a circuit via Trotter decomposition and compress it into a fixed depth circuit.

This method is not limited to creating a single particle. One can use the same unitary with different momentum $p$ to add another particle. Because $\ket{\psi_k}$ does not contain any particle with momentum $p \neq k$, $c_p \ket{\psi_k} = 0 $ still holds, and we can obtain
\begin{align}
    e^{i \frac{\pi}{2} \big( c^\dagger_p + c^{\phantom{\dagger}}_p  \big)}e^{i \frac{\pi}{2} \big( c^\dagger_k + c^{\phantom{\dagger}}_k  \big)} \ket{0} = i e^{i \frac{\pi}{2} \big( c^\dagger_p + c^{\phantom{\dagger}}_p  \big)} \ket{\psi_k} = - \big( c^\dagger_p + c^{\phantom{\dagger}}_p  \big) \ket{\psi_k} =  -c^\dagger_p \ket{\psi_k} = -c^\dagger_p \: c^\dagger_k \ket{0}.
\end{align}
This corresponds to first evolving under $\ham_k = c^\dagger_k + c^{\phantom{\dagger}}_k $, then under $\ham_p = c^\dagger_p + c^{\phantom{\dagger}}_p $. This is still a special case of one evolution under the time dependent Hamiltonian \cref{eq:lrfermions_c1}. Switching from $\ham_k$ to $\ham_p$ is just changing coefficients via simulation time. This can be applied for different creating more particles with different momenta as well. Thus we can create any number of particles with different momenta via using the compression of \cref{eq:lrfermions_c1}. 

As a remark, this is not limited to creating momentum definite states. This can be done for any orthonormal single particle basis simply by changing the coefficients of \cref{eq:momentum_creator_app}.

\section{\qblock\ Mappings for 1-D Spin Models with Periodic Boundary Condition}
\label{asec:pbc}

\subsection{TFIM with Periodic Boundary Condition}
\label{asec:PTFIM}
Consider the following TFIM Hamiltonian with periodic boundary condition
\begin{align}\label{aeq:TFIM_ham_pbc}
    \ham(t) = \sum_{i=1}^{n} \tilde{J}_i(t) X_i X_{i+1} + \sum_{i=1}^{n} \tilde{h}_i(t) Z_i
\end{align}
where $X_{n+1} := X_1$. We know that if it was not periodic but open boundary condition, the evolution under the Hamiltonian above would be compressed into a TFIM triangle \cite{kokcu2022algebraic, camps2022algebraic}. The periodic boundary condition term $X_n X_1$ prevents that, and cannot be represented via TFIM blocks. Instead, we provide the following \qblock\ mapping or this Hamiltonian
\begin{align}\label{eq:TFIM_PBC_qset}
\begin{split}
    Q^\PTFIM(\theta) &= e^{i \theta Z_2 Z_3 ... Z_n} = e^{ i \theta Z_1 P_Z}, \\
    B^\PTFIM_{2i-1}(\theta) &= e^{i \theta Z_i}, \\
    B^\PTFIM_{2i}(\theta) &= e^{i \theta X_{i} X_{i+1}},
\end{split}
\end{align}
where the superscript $\PTFIM$ stands for periodic transverse field Ising model, and $P_Z = Z_1 Z_2 \dots Z_n$ is the $Z$-parity operator.

Each term in the Hamiltonian can be represented via the mapping above. As it can be seen, apart from the boundary term $X_n X_1$, all other terms are already present in $B_i^\PTFIM$. The boundary term can be written as the following:
\begin{align}
    X_n X_1 = - (Y_1 Z_2 ... Z_{n-1} Y_n)\:P_Z.
\end{align}
After staring at it enough, one can recognize $Y_1 Z_2 ... Z_{n-1} Y_n$ as a mixture of hopping and pair creation/annihilation term between sites $1$ and $n$. From the results we obtained from fermionic swap operation, we know that this term can be written as a product of TFIM blocks. Considering that $P_Z$ commutes with all TFIM blocks, we find the following
\begin{align}\
\begin{split}
    X_n X_1 =& - \fswap_{n-1,n} ... \fswap_{2,3}\: Y_1 Y_2 \:\fswap_{2,3}^\dagger ... \fswap_{n-1,n}^\dagger \: P_Z \\
    =&  - \fswap_{n-1,n} ... \fswap_{2,3}\: Y_1 Y_2 \:P_Z \:\fswap_{2,3}^\dagger ... \fswap_{n-1,n}^\dagger \\
    =&  \fswap_{n-1,n} ... \fswap_{2,3}\: e^{i\frac{\pi}{4} Z_1 } e^{i\frac{\pi}{4} Z_2} 
    \:X_1 X_2 \:  P_Z \: e^{-i\frac{\pi}{4} Z_2 }e^{-i\frac{\pi}{4} Z_1 } \:\fswap_{2,3}^\dagger ... \fswap_{n-1,n}^\dagger, 
\end{split}
\end{align}
which is a product of the TFIM blocks and $Z_1 X_2 P_Z$. This new gate can be written as the following
\begin{align}
    X_1 X_2 P_Z = e^{i \frac{\pi}{4} Z_1} e^{i \frac{\pi}{4} X_1 X_2 }\: Z_1P_Z \: e^{-i \frac{\pi}{4} X_1 X_2 } e^{-i \frac{\pi}{4} Z_1 }. 
\end{align}
Thus, a rotation with the periodic boundary term $X_1 X_2$ can be written via the TFIM blocks, and rotation with $Z_1 P_Z$ which is $Q^\PTFIM$ itself. 

Let us show that the mapping \cref{eq:TFIM_PBC_qset} is a \qblock\ mapping. Since $B^\PTFIM_i = B^\TFIM_i$, we know that the block rules are already satisfied \cite{kokcu2022algebraic, camps2022algebraic}. Thus, we only need to show that the \qblock\ rules are satisfied. The $Q$-fusion rule is satisfied via the following relation 
\begin{align}
    Q^\PTFIM(\alpha) Q^\PTFIM(\beta) = Q^\PTFIM(\alpha + \beta).
\end{align}
$Q$-commutation is satisfied with $B^\PTFIM_{2i-1}$ because $[Z_i, Z_1 P_Z] = 0$, and with $B^\PTFIM_{2i}$ with $i \geq 1$ since $[X_i X_{i+1}, Z_1 P_Z] = [X_i X_{i+1}, Z_2 Z_3 \dots Z_n] = 0$. Finally, as it was the case for the TFIM \qblock mapping in \cref{eq:su2_relation}, $Q$-turnover property of the PTFIM mapping follows from the Euler decomposition of the following $\mathfrak{su}(2)$ generated by $Z_1 P_Z$ and $X_1 X_2$:
\begin{align}
    \mathfrak{su}(2) \equiv i \mathrm{span} \{ Z_1 P_Z, X_1 X_2, Y_1 X_2 P_Z\}.
\end{align} 
Therefore, via \cref{thm:q-compression}, time evolution of the Hamiltonian in \cref{aeq:TFIM_ham_pbc} can be compressed to a diamond.

In this form, the compression will lead to a circuit with $n$ qubits with $\bigO(n^2)$ depth and $\bigO(n^2)$ CNOT gates, due to the fact that $Q^\PTFIM(\theta)$ require $2n-4$ CNOTs and the same depth. With the observation $P_Z = Z_1 Z_2 ... Z_n$ commuting with every $B^\PTFIM_i$ and $Q^\PTFIM$, one can add one more qubit and reduce the circuit complexity into $\bigO(n)$ depth and $\bigO(n^2)$. To do so, one should put the $Z$-parity information into the added ancilla qubit, and apply controlled evolution via the ancilla by replacing $Q^\PTFIM(\theta)$ with 
\begin{align}
    Q^\TFIM(\theta) = e^{i \theta Z_0 Z_1}
\end{align}
where $0$ is the ancilla qubit. This is the Q-block we have introduced in \cref{eq:simple_controlled_qblocks}.

\subsection{TFXY Model with Periodic Boundary Condition}

Consider the following TFXY Hamiltonian with periodic boundary condition
\begin{align}\label{eq:TFXY_ham_pbc}
    \ham(t) = \sum_{i=1}^{n} \Big( \tilde{J}_i(t) X_i X_{i+1} + \tilde{K}_i(t) Y_i Y_{i+1} \Big) + \sum_{i=1}^{n} \tilde{h}_i(t) Z_i
\end{align}
where $X_{n+1} := X_1$ and $Y_{n+1} := Y_1$. We know that if it was not periodic but open boundary condition, the evolution under the Hamiltonian above would be compressed into a TFIM triangle \cite{kokcu2022algebraic, camps2022algebraic}. Here, we will show that the \qblock\ mapping given in \cref{eq:TFIM_PBC_qset} is capable of compressing the evolution of the periodic TFXY Hamiltonian above.

We already know that the mapping \cref{eq:TFIM_PBC_qset} is a \qblock\ mapping. The difference between the TFXY and the TFIM Hamiltonians is the $YY$ interaction. The following shows that the $YY$ terms can be generated via the terms given in TFIM:
\begin{align}
    Y_i Y_j = e^{i \frac{\pi}{4} Z_i } e^{i \frac{\pi}{4} Z_j } X_i X_j e^{-i \frac{\pi}{4} Z_i } e^{-i \frac{\pi}{4} Z_j }.
\end{align}
In a similar way, we can show that every term in the following ``generalized" TFXY model can be generated via the terms in the TFIM Hamiltonian:
\begin{align}\label{eq:gen_TFXY_ham_pbc}
\begin{split}
    \ham(t) = \sum_{i=1}^{n} \Big( &\tilde{J}_i(t) X_i X_{i+1} + \tilde{K}_i(t) X_i X_{i+1} \\
    +& \tilde{L}_i(t) X_i Y_{i+1} + \tilde{R}_i(t) Y_i X_{i+1} \Big) \\
    +& \sum_{i=1}^{n} \tilde{h}_i(t) Z_i,
\end{split}
\end{align}
where new $XY$ and $YX$ terms are also added. Since the periodic TFIM terms can be generated via the \qblock\ mapping \cref{eq:TFIM_PBC_qset}, then every term given in both Hamiltonians \cref{eq:TFXY_ham_pbc} and \cref{eq:gen_TFXY_ham_pbc} can be represented via the \qblock\ mapping \cref{eq:TFIM_PBC_qset} as well, and their evolution can be compressed into a diamond via \cref{thm:q-compression}. The properties of the resulting fixed depth circuit will be exactly the same as in \cref{asec:PTFIM}.

\section{\pblock s and Controlled Free Fermions}\label{subsec:pblocks}
Here, we introduce an new set of rules, referred to as the \pblock\ rules, which enable an efficient compression algorithm, serving as an alternative to the $Q$-block rules. These rules yield results comparable to those obtained with the \qblock\ rules discussed in \cref{sec:newblock_qblock}.
In this section, we will be explaining the \pblock\ rules, \pblock\ compression, and introduce a \pblock\ mapping for the controlled free fermions without the creation-annihilation operator. We discuss that this \pblock\ mapping can be considered as the TFXY version of the \qblock\ mapping introduced in \cref{subsec:controlled_free_fermions}\rredtext{, and has certain pros and cons over the controlled free fermionic $Q$-block mapping, e.g. $P$-block mapping can lead to a CNOT efficient circuit without the application of the TFIM $\rightarrow$ TFXY transformation given in \cref{eq:tfim2tfxy}, while $Q$-compression is based on TFIM compression, and can be implemented faster.}

\subsection{\pblock s and \pblock\ Compression}\label{subsec:pblock_rules_compression}

We define $P$-blocks and $P$-diamonds as follows:

\begin{definition}[$P$-Block]
\label{def:p-block}
Given blocks $B_i$ with $i \geq 1$, define a ``$P$-\emph{Block}'' $P=P(\myvec{\theta})$ as a structure that satisfies 
\begin{enumerate}
  \item \textbf{$P$-fusion:} For any set of parameters $\myvec{\alpha}$ and $\myvec{\beta}$, there exist $\myvec{a}$ such that
        \begin{equation}
          P(\myvec{\alpha}) \, P(\myvec{\beta}) = P(\myvec{a}),
        \end{equation}
  \item \textbf{$P$-commutation:} For any set of parameters $\myvec{\alpha}$ and $\myvec{\beta}$
        \begin{equation}
          P(\myvec{\alpha}) \, B_i(\myvec{\beta}) = B_i(\myvec{\beta}) \, P(\myvec{\alpha}), \qquad i>1,
        \end{equation}
  \item \textbf{$P$-turnover:} For any set of parameters $\myvec{\alpha}$, $\myvec{\beta}$, $\myvec{\gamma}$ and $\myvec{\theta}$  there exist $\myvec{a}$, $\myvec{b}$, $\myvec{c}$ and $\myvec{d}$ such that 
        \begin{equation}
          P(\myvec{\alpha}) \, B_1(\myvec{\beta}) \, P(\myvec{\gamma}) \, B_1(\myvec{\theta}) = B_1(\myvec{a}) \, P(\myvec{b}) \, B_1(\myvec{c}) \, P(\myvec{d}).
        \end{equation}
\end{enumerate}
If $P$ and blocks $B_i$ satisfy the properties listed above, we will say that $\{P,B_i\}$ is a $P$-block mapping.
\end{definition}

\begin{definition}[$P$-Diamond]
\label{def:p-diamond}
Define a ``$P$-diamond'' with height $n$ as
\begin{equation}
  D_n({\myvec{\alpha}},{\myvec{\beta}}) := \prod_{m=1}^{n+1} \left[ \left( \prod_{i=n \downarrow}^{1} B_i(\myvec{\alpha}_{i,m}) \right) P(\myvec{\beta_m}) \right]
\end{equation}
where $\downarrow$ in the product means that the multiplication is done in the decreasing order, and $B_i$ are blocks, $P$ is a P-block and each term in the product can have different variables.
\end{definition}

An illustration of the \pblock\ rules is given in \cref{fig:all_pblock_things} panels (a-c), and a $P$-diamond with height 4 is illustrated in \cref{fig:all_pblock_things} panel (d). Now we will prove that a $P$-diamond can absorb any block $B_i$ and a $P$-block.

\begin{theorem}[$P$-compression]\label{thm:p-compression}
    A $P$-diamond with height $n$ can be merged with any block $B_i$ with $i=1,2,\ldots,n$ and P-block $P$:
    \begin{align}
    \begin{split}
        &D_n({\myvec{\alpha}},{\myvec{\beta}})\,P(\myvec{\gamma}) = D_n({\myvec{a}},{\myvec{b}}), \\
        &D_n({\myvec{\alpha}},{\myvec{\beta}}) B_i(\myvec{\theta}) = D_n({\myvec{u}},{\myvec{v}}).
    \end{split}
    \end{align}
    Merging $P$ requires only one $P$-fusion. Merging $B_i$ requires $n+i-2$ $B$-turnovers, $1$ $P$-turnover and $1$ $B$-fusion.
\end{theorem}

\begin{proof}
The proof of this results is given diagramatically in \cref{fig:all_pblock_things} panel (e).
Since it is trivial to show that the $P$-diamond structure can absorb a $P$-block using the $P$-fusion operation,
we only discuss how the $P$-diamond can absorb $B_i$. The illustration is for $n = 4$ and $i=3$ case, but the same
mechanism generalizes to arbitrary (finite) system sizes and $B_i$.

In step 1, the block $B_i$ gets shifted upwards through the $P$-diamond via normal block turnover operations until its index reaches $1$.
Starting at index $i$ (which is 3 in the figure above), this upward movement requires $i-1$ turnover operations.
In step 2 a $P$-turnover is applied, which brings the block from the right side to the left side of the formation of blocks shown in the middle of the figure above.
Finally, in step 3, we can shift the block downward to the final row of blocks in the $P$-diamond.
This again requires only normal block turnover operations, more precisely, we need $n-1$ turnover operations to move the block down to the bottom row of
a $P$-diamond of height $n$ ($n=4$ in the figure above).
At this stage, the block can then be fused with corresponding index $n$ block. 
This operation requires $n+i-2$ turnover operations, $1$ $P$-turnover operation and $1$ $B$-fusion operation in total.
\end{proof}

In the next section we discuss controlled free fermionic Hamiltonians as an example of a
$P$-compressible Hamiltonian.
However, we note that the theorem does not directly rely on the Hamiltonian, it only requires a set of circuit elements that follow the $B$-block and $P$-block rules.
Thus, any circuit that is solely composed of gates that satisfy these rules can be compressed to a $P$-diamond, whether it is a controlled time evolution circuit of a free fermionic system or not.

\begin{figure*}[t]
    \centering
    \includegraphics[width = 0.75 \columnwidth]{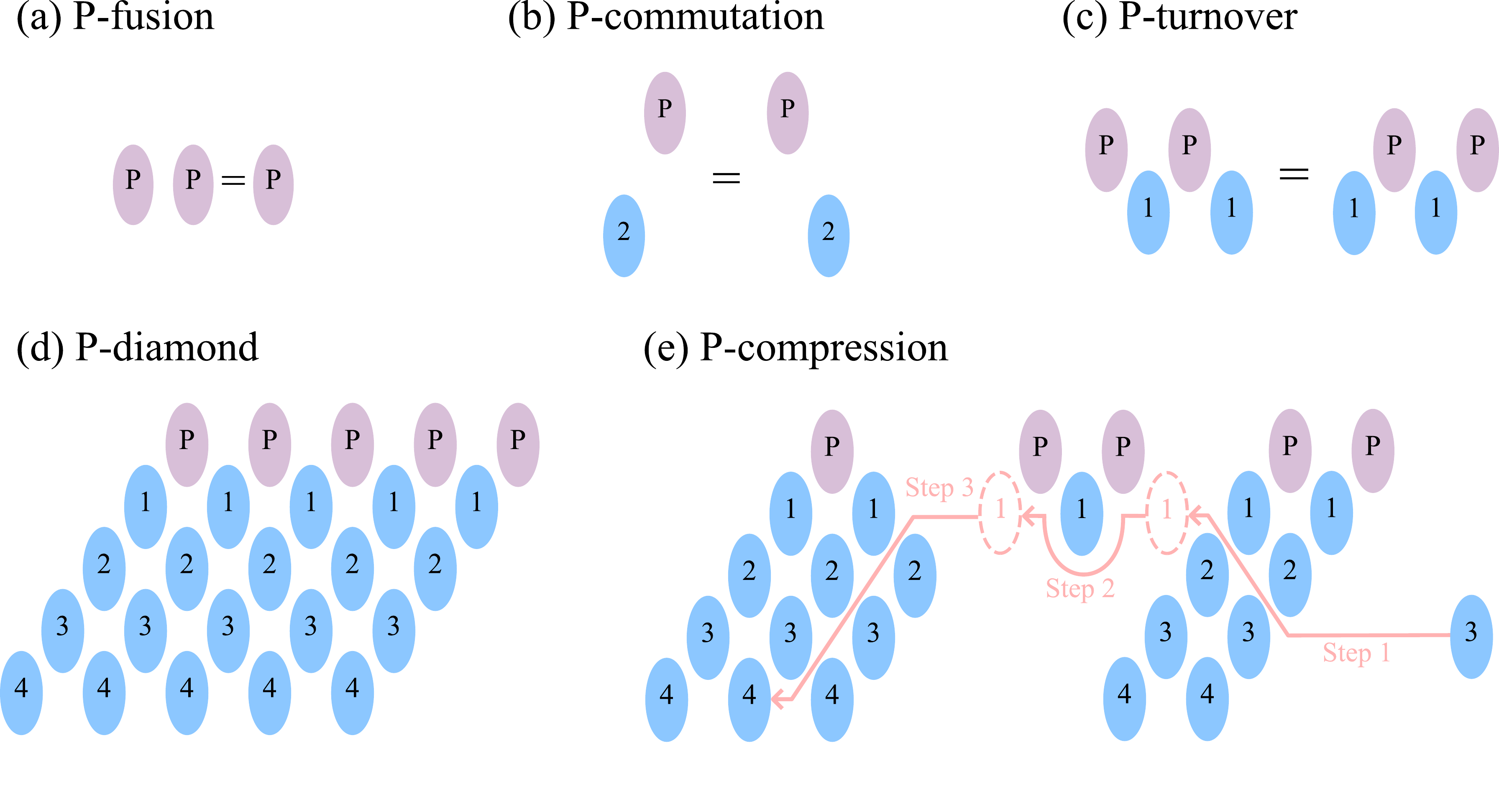}
    \caption{(a-c) $P$-block properties given in \cref{def:p-block}. Notice that the $P$-turnover is different from both turnover and $Q$-turnover rules, as it includes 4 blocks rather than 3 blocks. (d) The $P$-diamond structure defined in \cref{def:p-diamond} with heights $n=4$. (e) Proof of \cref{thm:p-compression}, i.e. how the $P$-diamond can absorb a block. 
    For this specific illustration, size of the $P$-diamond is $n=4$ and the index of the block is $i=3$. The block is lifted by other blocks via $B$-turnover operations, then passes through via $P$-turnover operation, and gets pushed down via $B$-turnover operations finally to be absorbed via a $B$-fusion operation. 
    }
    \label{fig:all_pblock_things}
\end{figure*}

\subsection{\pblock\ mapping for Controlled Free Fermions}\label{subsec:pblock_controlled_free_fermions}

In \cref{subsec:controlled_free_fermions}, we provided a $Q$-block mapping for controlled free fermionic Hamiltonian given in \cref{eq:simple_controlled}, and proved that the TFIM blocks and the gate $e^{i\theta Z_0 Z_1}$ was sufficient to represent time evolution of controlled free fermions. Here, we will show that the gate $e^{i\theta Z_0 Z_1}$ satisfies $P$-block rules if we consider a TFXY block mapping given in \cref{fig:all_block_things}(g) and \cref{eq:tfxy_block_mapping}, i.e. we will show that the following is a $P$-block mapping:
\begin{align}\label{eq:tfxy-pblock-mapping}
\begin{split}
    P^\TFXY(\theta) &= e^{i\theta Z_0 Z_1}, \\
    B^\TFXY_{i}(\myvec{\theta}) &= e^{i \theta_1\: Z_i}e^{i \theta_2\: Z_{i+1}}e^{i \theta_3\: X_i X_{i+1}} e^{i \theta_4\: Y_i Y_{i+1}}\:e^{i \theta_5\: Z_i}e^{i \theta_6\: Z_{i+1}}.
\end{split}
\end{align}

The $P^\TFXY$ satisfies the $P$-fusion due to the following:
\begin{align}
    P^\TFXY(\alpha) P^\TFXY(\beta) =P^\TFXY(\alpha + \beta). 
\end{align}
$P$-commutation is satisfied because for any $i>1$, $P^\TFXY$ and $B_i^\TFXY$ act on different sets of qubits.

The left-hand side of the $P$-turnover equation with the $P$-block mapping \cref{eq:tfxy-pblock-mapping} can be rewritten in terms of TFIM elements by transforming the TFXY block (i.e. a free fermionic matchgate) into TFIM elements by using the transformation given in \cref{eq:tfxy2tfim}. The entire expression can then be written as a sequence of $Q^\TFIM, B_1^\TFIM, B_2^\TFIM,$ and $B_3^\TFIM$. Thus, it can be compressed into a $Q$-diamond represented by the mapping given in \cref{eq:simple_controlled_qblocks}. Applying the grouping given in \cref{fig:controlled_TFIM_blocks}(e) to this $Q$-diamond yields the right-hand side of the $P$-turnover equation, proving that the $P$-turnover equality is satisfied in one direction. We can prove it in the other direction by simply observing the symmetry of the $Q$-block mapping rules, and applying the above procedure in reverse, i.e., compressing to a reverse $Q$-diamond. Thus, the mapping \cref{eq:tfxy-pblock-mapping} satisfies the $P$-turnover property as well, and is a $P$-block mapping.

The $P$-block mapping given in \cref{eq:tfxy-pblock-mapping} has its advantages and disadvantages over the $Q$-block maping for the controlled free fermions given in \cref{eq:simple_controlled_qblocks}. 
By using this $P$-block mapping, one can compress the controlled free fermionic evolution simply by using the TFXY blocks, and TFXY block rules given in \cite{kokcu2022algebraic, camps2022algebraic}, and can obtain a CNOT efficient circuit without the application of the TFIM $\leftrightarrow$ TFXY transformations given in \cref{eq:tfim2tfxy,eq:tfxy2tfim}. Thus, for devices that uses CNOT gates, it is more natural to use matchgates i.e. the TFXY blocks, and $P^\TFXY$ directly to implement the compression for the controlled free fermionic evolution. However, the $Q$-block mapping turns out to be more fundamental as the TFIM mapping, and faster to implement compression. In addition, $Q$-blocks can be generalized to controlled free fermions with creation as given in \cref{subsec:controlled_free_fermions_creation}.

\end{document}